\documentclass[draftcls,onecolumn,12pt]{IEEEtran}

\usepackage{cite}
\usepackage{graphicx,amssymb,amstext,amsmath,amsthm,equation,dsfont}
\usepackage{bm}
\usepackage{setspace}
\usepackage{stfloats}
\usepackage{flushend}
\usepackage{url}
\usepackage{array}
\usepackage{mdwmath}
\usepackage{mdwtab}
\usepackage{fixltx2e}
\usepackage{booktabs}
\usepackage{pstricks}
\usepackage{psfrag}
\usepackage{subfig}
\usepackage{hyperref}

\newtheorem{theorem}{Theorem}
\newtheorem{lemma}{Lemma}

\newtheorem{corollary}{Corollary}
\newtheorem{remark}{Remark}

\newcommand{\yv}{{\bm y}}

\newcommand{\yrv}{\bm Y}

\newcommand{\xv}{{\bm x}}

\newcommand{\xrv}{\bm X}

\newcommand{\zv}{{\bm z}}

\newcommand{\zrv}{\bm Z}

\newcommand{\HM}{{\mathsf H}}

\newcommand{\HRM}{\mathbb{H}}

\newcommand{\ERM}{\mathbb{E}}
\newcommand{\ARM}{\mathbb{A}}

\newcommand{\IM}{{\mathsf I}}
\newcommand{\nt}{{n_{\rm t}}}
\newcommand{\nr}{{n_{\rm r}}}
\newcommand{\nts}{{n_{{\rm t},s}}}
\newcommand{\ntone}{{n_{{\rm t},1}}}
\newcommand{\nttwo}{{n_{{\rm t},2}}}
\newcommand{\SNR}{{\sf SNR}}

\newcommand{\gmi}{I^{\rm gmi}}

\newcommand{\trans}[1]{#1^{\textnormal{\textsf{\tiny T}}}} 
\newcommand{\ii}{{\sf i}}
\newcommand{\Expect}{\mathsf{E}}

\ifCLASSINFOpdf
\else
\fi

\DeclareMathAlphabet{\mathpzc}{OT1}{pzc}{m}{it}

\DeclareSymbolFont{lettersA}{U}{txmia}{m}{it}
 \DeclareMathSymbol{\real}{\mathord}{lettersA}{"92}
 \DeclareMathSymbol{\field}{\mathord}{lettersA}{"83} 
 \DeclareMathSymbol{\integ}{\mathord}{lettersA}{"9A}
 \DeclareMathSymbol{\naturals}{\mathord}{lettersA}{"8E}

\newcommand{\hermi}[1]{#1^{\dagger}} 

\makeatletter
\def\tagform@#1{\maketag@@@{\ignorespaces#1\unskip\@@italiccorr}}
\let\orgtheequation\theequation
\def\theequation{(\orgtheequation)}
\makeatother

\begin{document}
\title{Nearest Neighbor Decoding and Pilot-Aided Channel Estimation for Fading Channels}

\author{A. Taufiq Asyhari, Tobias Koch and Albert Guill{\'e}n i F{\`a}bregas

\thanks{The material in this paper was presented in part at the 2011 IEEE International Symposium on Information Theory (ISIT), Saint Petersburg, Russia, July 31--August 5, 2011 and the 49th Annual Allerton Conference on Communication, Control and Computing, Monticello, Illinois, USA, September 28--30, 2011.}

\thanks{A. T. Asyhari was with Intelligent Information and Communications Research Center, Microelectronics and Information Systems Research Center, National Chiao Tung University, Hsinchu, Taiwan. He is now with School of Engineering and Informatics, University of Bradford, Bradford BD7 1DP, United Kingdom  (email: taufiq-a@ieee.org).}

\thanks{T. Koch was with the Department of Engineering, University of Cambridge, Cambridge, CB2 1PZ, United Kingdom. He is now with the Signal Theory and Communications Department, Universidad Carlos III de Madrid, 28911 Legan\'es, Spain (email: koch@tsc.uc3m.es).}

\thanks{A. Guill\'en i F\`abregas is with the Instituci\'o Catalana de Recerca i Estudis Avan\c{c}ats (ICREA), the Department of Information and Communication Technologies, Universitat Pompeu Fabra, Barcelona, Spain, and the Department of Engineering, University of Cambridge, Cambridge, CB2 1PZ, United Kingdom (email: guillen@ieee.org).}

\thanks{The work of A. T. Asyhari was supported in part by the Yousef Jameel Scholarship at University of Cambridge. T. Koch has received funding from the European's Seventh Framework Programme (FP7/2007--2013) under grant agreement No.\ 252663.}

}

\maketitle

\begin{abstract}
We study the information rates of non-coherent, stationary, Gaussian, multiple-input multiple-output (MIMO) flat-fading channels that are achievable with nearest neighbor decoding and pilot-aided channel estimation. In particular, we investigate the behavior of these achievable rates in the limit as the signal-to-noise ratio (SNR) tends to infinity by analyzing the capacity pre-log, which is defined as the limiting ratio of the capacity to the logarithm of the SNR as the SNR tends to infinity. We demonstrate that a scheme estimating the channel using pilot symbols and detecting the message using nearest neighbor decoding (while assuming that the channel estimation is perfect) essentially achieves the capacity pre-log of non-coherent multiple-input single-output flat-fading channels, and it essentially achieves the best so far known lower bound on the capacity pre-log of non-coherent MIMO flat-fading channels. We then extend our analysis to the multiple-access channel.
\end{abstract}


\IEEEpeerreviewmaketitle

\graphicspath{{Figure/EPS/}{Figure/}}

\section{Introduction}

The capacity of coherent multiple-input multiple-output (MIMO) channels increases with the signal-to-noise ratio (SNR) as $\min(\nt, \nr) \log \SNR$, where $\nt$ and $\nr$ are the number of transmit and receive antennas, respectively, and $\SNR$ denotes  the SNR per receive antenna \cite{Bell_foschini_layered_space-time,telatar_multiantenna_Gaussian}. The growth factor $\min(\nt, \nr)$ is sometimes referred to as the capacity pre-log \cite{IEEE:lapidoth:ontheasymptotic-capacity} or spatial multiplexing gain \cite{IEEE:zheng:diversityand}. This capacity growth can be achieved using a nearest neighbor decoder which selects the codeword that is closest (in a Euclidean distance sense) to the channel output. In fact, for coherent fading channels with additive Gaussian noise, this decoder is the maximum-likelihood decoder and is therefore optimal in the sense that it minimizes the error probability (see \cite{IEEE:lapidoth:nearestneighbournongaussian} and references therein). The coherent channel model assumes that there is a genie that provides the fading coefficients to the decoder; this assumption is difficult to achieve in practice. In this paper, we replace the role of the genie by a scheme that estimates the fading via pilot symbols. This can be viewed as a particular coding strategy over a non-coherent fading channel, i.e., a channel where both communication ends do not have access to fading coefficients but may be aware of the fading statistics. Note that with imperfect fading estimation, the nearest neighbor decoder that treats the fading estimate as if it were perfect is not necessarily optimal. Nevertheless, we show that, in some cases, nearest neighbor decoding with pilot-aided channel estimation achieves the capacity pre-log of non-coherent fading channels. (The capacity pre-log is defined as the limiting ratio of the capacity to the logarithm of the SNR as the SNR tends to infinity.)

The capacity of non-coherent fading channels has been studied in a number of works. Building upon \cite{Allerton:Marzetta:BLAST}, Hassibi and Hochwald \cite{IEEE:hassibi:howmuchtraining} studied the capacity of the block-fading channel and used pilot symbols (also known as training symbols) to obtain reasonably accurate fading estimates. Jindal and Lozano \cite{IEEE:Jindal:unifiedtreatmentblock-continuous-fading} provided tools for a unified treatment of pilot-based channel estimation in both block and stationary fading channels with bandlimited power spectral densities. In these works, lower bounds on the channel capacity were obtained. Lapidoth \cite{IEEE:lapidoth:ontheasymptotic-capacity} studied a single-input single-output (SISO) fading channel for more general stationary fading processes and showed that, depending on the predictability of the fading process, the capacity growth in SNR can be, \emph{inter alia}, logarithmic or double logarithmic. The extension of \cite{IEEE:lapidoth:ontheasymptotic-capacity} to multiple-input single-output (MISO) fading channels can be found in \cite{IEEE:koch:fadingnumber_degreeoffreedom}. A lower bound on the capacity of stationary MIMO fading channels was derived by Etkin and Tse in \cite{IEEE:etkin:degreeofffreedomMIMO}.

Lapidoth and Shamai \cite{IEEE:lapidoth:fadingchannels_howperfect} and Weingarten \emph{et. al.} \cite{IEEE:weingarten:gaussiancodes} studied non-coherent stationary fading channels from a mismatched-decoding perspective. In particular, they studied achievable rates with Gaussian codebooks and nearest neighbor decoding. In both works, it is assumed that there is a genie that provides imperfect estimates of the fading coefficients.

In this work, we add the estimation of the fading coefficients to the analysis. In particular, we study a communication system where the transmitter emits pilot symbols at regular intervals, and where the receiver separately performs \emph{channel estimation} and \emph{data detection}. Specifically, based on the channel outputs corresponding to pilot transmissions, the channel estimator produces estimates of the fading for the remaining time instants using a linear minimum mean-square error (LMMSE) interpolator. Using these estimates, the data detector employs a nearest neighbor decoder that detects the transmitted message. We study the achievable rates of this communication scheme at high SNR. In particular, we study the pre-log for fading processes with bandlimited power spectral densities. (The pre-log is defined as the limiting ratio of the achievable rate to the logarithm of the SNR as the SNR tends to infinity.)

For SISO fading channels, using some simplifying arguments, Lozano \cite{IEEE:lozano:interplay_spectral_doppler} and Jindal and Lozano \cite{IEEE:Jindal:unifiedtreatmentblock-continuous-fading} showed that this scheme achieves the capacity pre-log. In this paper, we prove this result without any simplifying assumptions and extend it to MIMO fading channels. We show that the maximum rate pre-log with nearest neighbor decoding and pilot-aided channel estimation is given by the capacity pre-log of the coherent fading channel $\min(\nt,\nr)$ times the fraction of time used for the transmission of data. Hence, the loss with respect to the coherent case is solely due to the transmission of pilots used to obtain accurate fading estimates. If the inverse of twice the bandwidth of the fading process is an integer, then for MISO channels, the above scheme achieves the capacity pre-log derived by Koch and Lapidoth \cite{IEEE:koch:fadingnumber_degreeoffreedom}. For MIMO channels, the above scheme achieves the best so far known lower bound on the capacity pre-log obtained in \cite{IEEE:etkin:degreeofffreedomMIMO}.

The rest of the paper is organized as follows. Section~\ref{sec:model} describes the channel model and introduces our transmission scheme along with nearest neighbor decoding and pilots for channel estimation. Section~\ref{sec:prelog} defines the pre-log and presents the main result. Section~\ref{sec:mac} extends the use of our scheme to a fading multiple-access channel (MAC). Sections \ref{sec:proof-theorem-1} and \ref{sec:proof-mac-pre-log} provide the proofs of our main results. Section \ref{sec:conclusion} summarizes the results and concludes the paper.

\section{System Model and Transmission Scheme}
\label{sec:model}

We consider a discrete-time MIMO flat-fading channel with $\nt$ transmit antennas and $\nr$ receive antennas. Thus, the channel output at time instant $k \in \integ$ (where $\integ$ denotes the set of integers) is the complex-valued $\nr$-dimensional random vector given by 
\begin{equation}
 \yrv_k = \sqrt{\frac{\sf SNR}{\nt}} \HRM_k \xv_k + \zrv_k. \label{eq:channelmodel}
\end{equation}
Here  $\xv_k \in \field^{\nt}$ denotes the time-$k$ channel input vector (with $\field$ denoting the set of complex numbers), $\HRM_k$ denotes the $(\nr \times \nt)$-dimensional random fading matrix at time $k$, and $\zrv_k $ denotes the $\nr$-variate random additive noise vector at time $k$.

The noise process $\{\zrv_k, k \in \integ \}$ is a sequence of independent and identically distributed (i.i.d.) complex-Gaussian random vectors with zero mean and covariance matrix ${\sf I}_{\nr}$, where ${\sf I}_{\nr}$ is the $\nr \times \nr$ identity matrix. $\sf SNR$ denotes the average SNR for each received antenna.

The fading process $\{ \HRM_k, k \in \integ \}$ is stationary, ergodic and complex-Gaussian. We assume that the $\nr\cdot\nt$ processes $\{H_k(r,t), k \in \integ \}$, $r=1,\ldots,\nr$, $t=1,\ldots,\nt$ are independent and have the same law, with each process having zero mean, unit variance, and power spectral density $f_H (\lambda)$, $-\frac{1}{2} \leq \lambda \leq \frac{1}{2}$. Thus, $f_H(\cdot)$ is a non-negative (measurable) function satisfying
 \begin{equation}
 \mathsf{E} \left[ H_{k+m} (r,t) H^*_{k} (r,t) \right] = \int^{1/2}_{-1/2} e^{\ii 2\pi m \lambda} f_H (\lambda) d \lambda,
\end{equation}
where $(\cdot)^*$ denotes complex conjugation, and where $\ii \triangleq \sqrt{-1}$. We further assume that the power spectral density $f_H(\cdot)$ has bandwidth $\lambda_D<1/2$, i.e., $f_H (\lambda) = 0$ for $|\lambda| > \lambda_D$ and $f_H(\lambda)>0$ otherwise. We finally assume that the fading process $\{\HRM_k, k\in\integ\}$ and the noise process $\{\zrv_k, k \in \integ \}$ are independent and that their joint law does not depend on $\{\xv_k, k\in\integ\}$.

The transmission involves both codewords and pilots. The former conveys the message to be transmitted, and the latter are used to facilitate the estimation of the fading coefficients at the receiver. We denote a codeword conveying a message $m$, $m \in \mathcal{M}$ (where $\mathcal{M} = \left \{1,\dotsc, \lfloor e^{nR} \rfloor \right \}$ is the set of possible messages, and where $\lfloor b \rfloor$ denotes the largest integer smaller than or equal to $b$) at rate $R$ by the length-$n$ sequence of input vectors $\bar \xv_1 {(m)},\dotsc,\bar \xv_n {(m)}$. The codeword is selected from the codebook $\mathcal{C}$, which is drawn i.i.d. from an $\nt$-variate complex-Gaussian distribution with zero mean and identity covariance matrix such that 
\begin{equation}
\frac{1}{n} \sum^{n}_{k=1}\mathsf{E} \left[ \left \| \bar \xrv_k {(m)} \right \|^2 \right] = \nt,~~ m \in \mathcal{M}  \label{eq:pw_constraint}
\end{equation} 
where $\|\cdot\|$ denotes the Euclidean norm.

To estimate the fading matrix, we transmit orthogonal pilot vectors. The pilot vector ${\bm p}_t \in \field^\nt$ used to estimate the fading coefficients corresponding to the $t$-th transmit antenna is given by $p_t(t)=1$ and $p_t(t')=0$ for $t'\neq t$. For example, the first pilot vector is ${\bm p}_1=\trans{\left(1,0,\cdots,0 \right)}$, where $\trans{(\cdot)}$ denotes the transpose. To estimate the whole fading matrix, we thus need to send the $\nt$ pilot vectors ${\bm p}_1,\ldots,{\bm p}_{\nt}$.

 \begin{figure*}[t]
 \begin{center}
\begin{pspicture}(-5,-0.65)(12,1.8)

\psframe[linewidth=0.02,fillstyle=solid,fillcolor=darkgray](-2,1.5)(-1.7,1.8)
\rput(-1.2,1.65){\scriptsize Pilot}

\psframe[linewidth=0.02,fillstyle=solid,fillcolor=lightgray](3.7,1.5)(4,1.8)
\rput(4.5,1.65){\scriptsize Data}

\psframe[linewidth=0.02](8.0,1.5)(8.3,1.8)
\rput(9.5,1.65){\scriptsize No transmission}

\rput(4.3,0.9){\scriptsize $N + \left(\frac{N}{L - n_{\rm t}} + 1 \right)n_{\rm t}$}
\psline[linewidth=0.02,linestyle=dotted,dotsep=1pt]{<->}(0,0.5)(8.6,0.5)

\rput(6.95,-1){\scriptsize $L$}
\psline[linewidth=0.02,linestyle=dotted,dotsep=1pt]{<->}(5.9,-0.8)(8,-0.8)

\rput(9.65,-1){\scriptsize $L(T-1)$}
\psline[linewidth=0.02,linestyle=dotted,dotsep=1pt]{<->}(8.6,-0.8)(10.7,-0.8)

\rput(-1.05,-1){\scriptsize $L(T-1)$}
\psline[linewidth=0.02,linestyle=dotted,dotsep=1pt]{<->}(-2.1,-0.8)(0,-0.8)

{
\rput(-3.5,0.15){\scriptsize $t=1$}

\psframe[linewidth=0.02,fillstyle=solid,fillcolor=darkgray](-2.1,0)(-1.8,0.3)
\psframe[linewidth=0.02](-1.8,0)(-1.5,0.3)
\psframe[linewidth=0.02](-1.5,0)(-1.2,0.3)
\psframe[linewidth=0.02](-1.2,0)(-0.9,0.3)
\psframe[linewidth=0.02](-0.9,0)(-0.6,0.3)
\psframe[linewidth=0.02](-0.6,0)(-0.3,0.3)
\psframe[linewidth=0.02](-0.3,0)(0,0.3)

\psframe[linewidth=0.02,fillstyle=solid,fillcolor=darkgray](0,0)(0.3,0.3)
\psframe[linewidth=0.02](0.3,0)(0.6,0.3) 
\psframe[linewidth=0.02,fillstyle=solid,fillcolor=lightgray](0.6,0)(0.9,0.3)
\psframe[linewidth=0.02,fillstyle=solid,fillcolor=lightgray](0.9,0)(1.2,0.3)
\psframe[linewidth=0.02,fillstyle=solid,fillcolor=lightgray](1.2,0)(1.5,0.3)
\psframe[linewidth=0.02,fillstyle=solid,fillcolor=lightgray](1.5,0)(1.8,0.3)
\psframe[linewidth=0.02,fillstyle=solid,fillcolor=lightgray](1.8,0)(2.1,0.3)

\psframe[linewidth=0.02,fillstyle=solid,fillcolor=darkgray](2.1,0)(2.4,0.3)
\psframe[linewidth=0.02](2.4,0)(2.7,0.3) 
\psframe[linewidth=0.02,fillstyle=solid,fillcolor=lightgray](2.7,0)(3,0.3)
\psframe[linewidth=0.02,fillstyle=solid,fillcolor=lightgray](3,0)(3.3,0.3)

\pscircle[fillstyle=solid,fillcolor=black,linewidth=0.0875,linecolor=white](3.8,0.15){0.14}
\pscircle[fillstyle=solid,fillcolor=black,linewidth=0.0875,linecolor=white](4.3,0.15){0.14}
\pscircle[fillstyle=solid,fillcolor=black,linewidth=0.0875,linecolor=white](4.8,0.15){0.14}

\psframe[linewidth=0.02](5.3,0)(8.6,0.3)
\psframe[linewidth=0.02,fillstyle=solid,fillcolor=lightgray](5.3,0)(5.6,0.3)
\psframe[linewidth=0.02,fillstyle=solid,fillcolor=lightgray](5.6,0)(5.9,0.3)
\psframe[linewidth=0.02,fillstyle=solid,fillcolor=darkgray](5.9,0)(6.2,0.3)
\psframe[linewidth=0.02](6.2,0)(6.5,0.3)
\psframe[linewidth=0.02,fillstyle=solid,fillcolor=lightgray](6.5,0)(6.8,0.3)
\psframe[linewidth=0.02,fillstyle=solid,fillcolor=lightgray](6.8,0)(7.1,0.3)
\psframe[linewidth=0.02,fillstyle=solid,fillcolor=lightgray](7.1,0)(7.4,0.3)
\psframe[linewidth=0.02,fillstyle=solid,fillcolor=lightgray](7.4,0)(7.7,0.3)
\psframe[linewidth=0.02,fillstyle=solid,fillcolor=lightgray](7.7,0)(8,0.3)

\psframe[linewidth=0.02,fillstyle=solid,fillcolor=darkgray](8,0)(8.3,0.3)
\psframe[linewidth=0.02](8.3,0)(8.6,0.3)
\psframe[linewidth=0.02](8.6,0)(8.9,0.3)
\psframe[linewidth=0.02](8.9,0)(9.2,0.3)
\psframe[linewidth=0.02](9.2,0)(9.5,0.3)
\psframe[linewidth=0.02](9.5,0)(9.8,0.3)
\psframe[linewidth=0.02](9.8,0)(10.1,0.3)

\psframe[linewidth=0.02,fillstyle=solid,fillcolor=darkgray](10.1,0)(10.4,0.3)
\psframe[linewidth=0.02](10.4,0)(10.7,0.3)

}

{
\rput(-3.5,-0.45){\scriptsize $t=2$}

\psframe[linewidth=0.02](-2.1,-0.6)(-1.8,-0.3)
\psframe[linewidth=0.02,fillstyle=solid,fillcolor=darkgray](-1.5,-0.6)(-1.8,-0.3)
\psframe[linewidth=0.02](-1.5,-0.6)(-1.2,-0.3) 
\psframe[linewidth=0.02](-1.2,-0.6)(-0.9,-0.3)
\psframe[linewidth=0.02](-0.9,-0.6)(-0.6,-0.3) 
\psframe[linewidth=0.02](-0.6,-0.6)(-0.3,-0.3)
\psframe[linewidth=0.02](-0.3,-0.6)(0,-0.3) 

\psframe[linewidth=0.02](0,-0.6)(0.3,-0.3)
\psframe[linewidth=0.02,fillstyle=solid,fillcolor=darkgray](0.3,-0.6)(0.6,-0.3)
\psframe[linewidth=0.02,fillstyle=solid,fillcolor=lightgray](0.6,-0.6)(0.9,-0.3)
\psframe[linewidth=0.02,fillstyle=solid,fillcolor=lightgray](0.9,-0.6)(1.2,-0.3)
\psframe[linewidth=0.02,fillstyle=solid,fillcolor=lightgray](1.2,-0.6)(1.5,-0.3)
\psframe[linewidth=0.02,fillstyle=solid,fillcolor=lightgray](1.5,-0.6)(1.8,-0.3)
\psframe[linewidth=0.02,fillstyle=solid,fillcolor=lightgray](1.8,-0.6)(2.1,-0.3)

\psframe[linewidth=0.02](2.1,-0.6)(2.4,-0.3)
\psframe[linewidth=0.02,fillstyle=solid,fillcolor=darkgray](2.4,-0.6)(2.7,-0.3)
\psframe[linewidth=0.02,fillstyle=solid,fillcolor=lightgray](2.7,-0.6)(3,-0.3)
\psframe[linewidth=0.02,fillstyle=solid,fillcolor=lightgray](3,-0.6)(3.3,-0.3)

\pscircle[fillstyle=solid,fillcolor=black,linewidth=0.0875,linecolor=white](3.8,-0.45){0.14}
\pscircle[fillstyle=solid,fillcolor=black,linewidth=0.0875,linecolor=white](4.3,-0.45){0.14}
\pscircle[fillstyle=solid,fillcolor=black,linewidth=0.0875,linecolor=white](4.8,-0.45){0.14}

\psframe[linewidth=0.02,fillstyle=solid,fillcolor=lightgray](5.3,-0.6)(5.6,-0.3)
\psframe[linewidth=0.02,fillstyle=solid,fillcolor=lightgray](5.6,-0.6)(5.9,-0.3)
\psframe[linewidth=0.02](5.9,-0.6)(6.2,-0.3)
\psframe[linewidth=0.02,fillstyle=solid,fillcolor=darkgray](6.2,-0.6)(6.5,-0.3)

\psframe[linewidth=0.02,fillstyle=solid,fillcolor=lightgray](6.5,-0.6)(6.8,-0.3)
\psframe[linewidth=0.02,fillstyle=solid,fillcolor=lightgray](6.8,-0.6)(7.1,-0.3)
\psframe[linewidth=0.02,fillstyle=solid,fillcolor=lightgray](7.1,-0.6)(7.4,-0.3)
\psframe[linewidth=0.02,fillstyle=solid,fillcolor=lightgray](7.4,-0.6)(7.7,-0.3)
\psframe[linewidth=0.02,fillstyle=solid,fillcolor=lightgray](7.7,-0.6)(8,-0.3)

\psframe[linewidth=0.02](8,-0.6)(8.3,-0.3)
\psframe[linewidth=0.02,fillstyle=solid,fillcolor=darkgray](8.3,-0.6)(8.6,-0.3)
\psframe[linewidth=0.02](8.6,-0.6)(8.9,-0.3)
\psframe[linewidth=0.02](8.9,-0.6)(9.2,-0.3)
\psframe[linewidth=0.02](9.2,-0.6)(9.5,-0.3)
\psframe[linewidth=0.02](9.5,-0.6)(9.8,-0.3)
\psframe[linewidth=0.02](9.8,-0.6)(10.1,-0.3)

\psframe[linewidth=0.02](10.1,-0.6)(10.4,-0.3)
\psframe[linewidth=0.02,fillstyle=solid,fillcolor=darkgray](10.7,-0.6)(10.4,-0.3)

}

\end{pspicture}
\end{center}
 \caption{Structure of pilot and data transmission for $n_{\rm t} = 2$, $L = 7$ and $T=2$.}
 \vspace*{-2mm}
 \label{fig:pilot_data_illustration}
 \end{figure*}
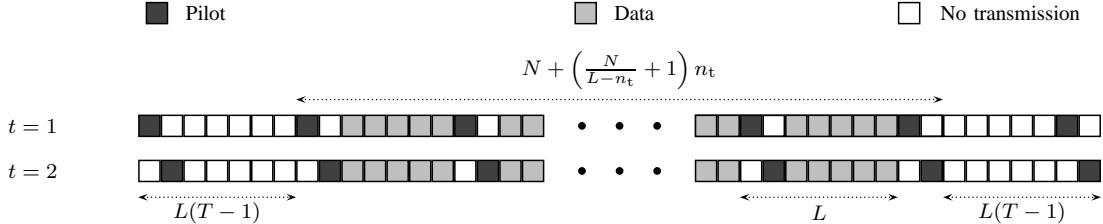

The transmission scheme is as follows. Every $L$ time instants (for some $L\in \naturals$, where $\naturals$ is the set of all positive integers), we transmit the $\nt$ pilot vectors ${\bm p}_1,\ldots,{\bm p}_{\nt}$. Each codeword is then split up into blocks of $L-\nt$ data vectors, which will be transmitted after the $\nt$ pilot vectors. The process of transmitting $L-\nt$ data vectors and $\nt$ pilot vectors continues until all $n$ data vectors are completed. Herein we assume that $n$ is an integer multiple of $L-\nt$.\footnote{If $n$ is not an integer multiple of $L-\nt$, then the last $L-\nt$ instants are not fully used by data vectors and contain therefore time instants where we do not transmit anything. The thereby incurred loss in information rate vanishes as $n$ tends to infinity.} Prior to transmitting the first data block, and after transmitting the last data block, we introduce a guard period of $L (T -1)$ time instants (for some $T\in\naturals$), where we transmit every $L$ time instants the $\nt$ pilot vectors ${\bm p}_1,\ldots,{\bm p}_{\nt}$, but we do not transmit data vectors in between. The guard period ensures that, at every time instant, we can employ a channel estimator that bases its estimation on the channel outputs corresponding to the $T$ past and the $T$ future pilot transmissions. This facilitates the analysis and does not incur any loss in terms of achievable rates. The above transmission scheme is illustrated in Fig.~\ref{fig:pilot_data_illustration}. The channel estimator is described in the following. 

Note that the total block-length of the above transmission scheme (comprising data vectors, pilot vectors and guard period) is given by
\begin{equation}
n' = n_{\rm p} + n + n_{\rm g} \label{eq:total_length}
\end{equation}
where $n_{\rm p}$ denotes the number of channel uses reserved for pilot vectors, and where $n_{\rm g}$ denotes the number of channel uses during the silent guard period, i.e.,
\begin{align}{}
 n_{\rm p} &= \left(\frac{n}{L - \nt} + 1  + 2 (T-1) \right) \nt, \label{eq:number-of-pilots} \\
 n_{\rm g} &= 2(L-\nt)(T-1). \label{eq:number-for-guard}
\end{align}

We now turn to the decoder. Let $\mathcal{D}$ denote the set of time indices where data vectors of a codeword are transmitted, and let $\mathcal{P}$ denote the set of time indices where pilots are transmitted. The decoder consists of two parts: a \emph{channel estimator} and a \emph{data detector}. The channel estimator considers the channel output vectors ${\yrv}_k$, $k\in\mathcal{P}$ corresponding to the past and future $T$ pilot transmissions and estimates $H_{k}(r,t)$ using a linear interpolator, so the estimate $\hat{H}_k^{(T)}(r,t)$ of the fading coefficient $H_k(r,t)$ is given by
\begin{equation}
\label{eq:LMMSEestimation}
 \hat H_k^{(T)}(r,t) = \sum^{k + T L}_{\substack{ k' = k - TL:\\ k' \in \mathcal{P}}} a_{k'} (r,t) Y_{k'}(r)
\end{equation}
where the coefficients $a_{k'} (r,t)$ are chosen in order to minimize the mean-squared error.\footnote{It has been shown in \cite{IEEE:ohno:averageratePSAM} that for the linear interpolator in \eqref{eq:LMMSEestimation}, only the observations when pilots are transmitted, $\yrv_{k'},~k' \in \mathcal{P}$ are relevant for fading estimation. }

Note that, since the pilot vectors transmit only from one antenna, the fading coefficients corresponding to all transmit and receive antennas $(r,t)$ can be observed. Further note that, since the fading processes $\{H_k(r,t), k \in \integ \}$, $r=1,\ldots,\nr$, $t=1,\ldots,\nt$ are independent, estimating $H_{k}(r,t)$ only based on $\{Y_k(r),k\in\integ\}$ rather than on $\{\yrv_k,k\in\integ\}$ incurs no loss in optimality.

Since the time-lags between $\HRM_k$, $k\in\mathcal{D}$ and the observations $\yrv_{k'}$, $k'\in\mathcal{P}$ depend on $k$, it follows that the interpolation error
\begin{equation}
E_k^{(T)}(r,t)  \triangleq  H_k(r,t) - \hat H_k^{(T)}(r,t) \label{eq:Ert}
\end{equation}
is not stationary but cyclo-stationary with period $L$. It can be shown that, irrespective of $r$, the variance of the interpolation error
\begin{equation}
\epsilon^{2}_{\ell,T}(r,t) \triangleq \mathsf{E}\left[\left|H_k(r,t)-\hat{H}^{(T)}_k(r,t)\right|^2\right]
\end{equation}
tends to the following expression as $T$ tends to infinity \cite{IEEE:ohno:averageratePSAM}
\begin{align}
 \epsilon^2_\ell (t) & \triangleq  \lim_{T \rightarrow \infty} \epsilon^2_{\ell,T}(r,t) \\ 
 &= 1 -  \int^{1/2}_{-1/2} \frac{\SNR |f_{L,\ell-t+1}(\lambda) |^2}{\SNR f_{L,0} (\lambda) + \nt} d \lambda \label{eq:interpolation-error-variance-T-infinity-general}
\end{align}
where $\ell \triangleq k\mod L$ denotes the remainder of $k/L$. Here $f_{L,\ell}(\cdot)$ is given by
\begin{equation}
 f_{L,\ell} (\lambda) = \frac{1}{L} \sum^{L-1}_{\nu=0} \bar{f}_H \left( \frac{\lambda - \nu}{L} \right)e^{\ii 2\pi \ell \frac{\lambda - \nu}{L} }, \qquad \ell =0,\dotsc,L-1 \label{eq:undersampled-spectrum-ell}
\end{equation}
 and $\bar{f}_H(\cdot)$ is the periodic continuation of $f_H (\cdot)$, i.e., it is the periodic function of period $[-1/2,1/2)$ that coincides with $f_H(\lambda)$ for $-1/2\leq\lambda\leq 1/2$. If
\begin{equation}
 L \leq \frac{1}{2 \lambda_D}  \label{eq:nyquist}
\end{equation}
then $|f_{L,\ell}(\cdot)|$ becomes
\begin{equation}
|f_{L,\ell}(\lambda)| = f_{L,0}(\lambda)=\frac{1}{L} f_H\left(\frac{\lambda}{L}\right), ~~ -\frac{1}{2}\leq\lambda\leq\frac{1}{2}.
\end{equation}
In this case, irrespective of $\ell$ and $t$, the variance of the interpolation error is given by
\begin{equation}
 \epsilon^2_\ell (t) = \epsilon^2 = 1 -  \int^{1/2}_{-1/2} \frac{\SNR \left [f_{H}(\lambda) \right]^2}{\SNR f_{H} (\lambda) + L\nt} d \lambda, \qquad \ell =0,\dotsc,L-1,~t = 1,\dotsc,\nt \label{eq:fadingestimate_error}
\end{equation}
which vanishes as the $\SNR$ tends to infinity. Recall that $\lambda_D$ denotes the bandwidth of $f_H(\cdot)$. Thus, \eqref{eq:nyquist} implies that no aliasing occurs as we undersample the fading process $L$ times. Note that in contrast to \eqref{eq:interpolation-error-variance-T-infinity-general}, the variance in \eqref{eq:fadingestimate_error} is independent of the transmit antenna index $t$. See Section \ref{subsec:channelestimator} for a more detailed discussion.

The channel estimator feeds the sequence of fading estimates $\{\hat\HRM_k^{(T)},k\in\mathcal{D}\}$ (which is composed of the matrix entries $\{\hat{H}_k^{(T)}(r,t),k\in\mathcal{D}\}$) to the data detector. We shall denote its realization by $\{\hat \HM^{(T)}_k, k \in \mathcal{D} \}$. Based on the channel outputs $\{\yv_k,k\in \mathcal{D} \}$ and fading estimates $\{\hat\HM_k^{(T)},k\in\mathcal{D}\}$, the data detector uses a nearest neighbor decoder to guess which message was transmitted. Thus, the decoder decides on the message $\hat m$ that satisfies
\begin{equation}
 \hat m = \arg \min_{m \in\mathcal{M}}  D(m) 
\end{equation}
where
\begin{equation}
 D(m) \triangleq \sum_{k \in \mathcal{D}^{(n')}} \left \|\yv_k - \sqrt{ \frac{\SNR}{ \nt}} ~ \hat \HM^{(T)}_k \xv_k {(m)} \right\|^2. \label{eq:D-m-secII}
\end{equation}
On the RHS of \eqref{eq:D-m-secII},  assuming that the first pilot symbol is transmitted at time $k=0$, we have defined 
\begin{equation}
\mathcal{D}^{(n')} \triangleq  \{0,\dotsc, n' - 1 \} \cap \mathcal{D} \label{eq:Dnprime-defined}
\end{equation}
as a set of time indices for a single codeword transmission.

\section{The Pre-Log}
\label{sec:prelog}

We say that a rate
\begin{equation}
R(\SNR) \triangleq \frac{\log |\mathcal{M}|}{n}
\end{equation}
is achievable if there exists a code with $\lfloor e^{nR} \rfloor$ codewords such that the error probability tends to zero as the codeword length $n$ tends to infinity. In this work, we study the set of rates that are achievable with nearest neighbor decoding and pilot-aided channel estimation. We focus on the achievable rates at high $\SNR$. In particular, we are interested in the maximum achievable pre-log, defined as
\begin{equation}
\Pi_{R^*} \triangleq \limsup_{\SNR \rightarrow \infty} ~ \frac{R^*(\SNR)}{\log \SNR}  \label{def:gmi-pre-log-gaussian}
\end{equation}
where $R^*(\SNR)$ is the maximum achievable rate, maximized over all possible encoders.

The capacity pre-log---which is given by \eqref{def:gmi-pre-log-gaussian} but with $R^*(\SNR)$ replaced by the capacity\footnote{The capacity is defined as the supremum of all achievable rates maximized over all possible encoders and decoders.} $C(\SNR)$---of SISO fading channels was computed by Lapidoth \cite{IEEE:lapidoth:ontheasymptotic-capacity} as
\begin{equation}
\Pi_{C} = \mu \big(\{\lambda\colon f_H (\lambda) = 0\} \big) \label{eq:eta}
\end{equation}
where $\mu (\cdot)$ denotes the Lebesgue measure on the interval $[-1/2,1/2]$. Koch and Lapidoth \cite{IEEE:koch:fadingnumber_degreeoffreedom} extended this result to MISO fading channels and showed that if the fading processes $\{H_k(t),k\in\integ\}$, $t=1,\ldots,\nt$ are independent and have the same law, then the capacity pre-log of MISO fading channels is equal to the capacity pre-log of the SISO fading channel with fading process $\{H_k(1),k\in\integ\}$. Using \eqref{eq:eta}, the capacity pre-log of MISO fading channels with bandlimited power spectral densities of bandwidth $\lambda_D$ can be evaluated as
\begin{equation}
\label{eq:MISO-capacity-prelog-BL}
\Pi_C = 1-2\lambda_D.
\end{equation}
Since $R^*(\SNR)\leq C(\SNR)$, it follows that $\Pi_{R^*}\leq\Pi_C$.

To the best of our knowledge, the capacity pre-log of MIMO fading channels is unknown. For independent fading processes $\{H_k(r,t),k\in\integ\}$, $t=1,\dots,\nt$, $r=1,\ldots,\nr$ that have the same law, the best so far known lower bound on the MIMO pre-log is due to Etkin and Tse \cite{IEEE:etkin:degreeofffreedomMIMO}, and is given by
\begin{equation}
\Pi_{C} \geq \min (\nt, \nr) \Big(1 - \min(\nt,\nr) \mu\big(\{ \lambda\colon f_H(\lambda)> 0\}\big) \Big).\label{eq:etaMIMO}
\end{equation}
For power spectral densities that are bandlimited to $\lambda_D$, this becomes
\begin{equation}
\label{eq:MIMO-capacity-prelog-BL}
\Pi_C \geq \min (\nt, \nr) \big(1 - \min(\nt,\nr)\, 2\lambda_D\big).
\end{equation}
Observe that \eqref{eq:MIMO-capacity-prelog-BL} specializes to \eqref{eq:MISO-capacity-prelog-BL} for $\nr=1$. It should be noted that the capacity pre-log for MISO and SISO fading channels was derived under a peak-power constraint on the channel inputs, whereas the lower bound on the capacity pre-log for MIMO fading channels was derived under an average-power constraint. Clearly, the capacity pre-log corresponding to a peak-power constraint can never be larger than the capacity pre-log corresponding to an average-power constraint. It is believed that the two pre-logs are in fact identical (see the conclusion in \cite{IEEE:lapidoth:ontheasymptotic-capacity}).

In this paper, we show that a communication scheme that employs nearest neighbor decoding and pilot-aided channel estimation achieves the following pre-log.

\begin{theorem}
\label{th:pre-log-MIMO}
 Consider the Gaussian MIMO flat-fading channel with $\nt$ transmit antennas and $\nr$ receive antennas \eqref{eq:channelmodel}. Then, the transmission and decoding scheme described in Section \ref{sec:model} achieves
\begin{equation}
 \Pi_{R^*} \geq \min(\nt, \nr) \left(1 - \frac{\min(\nt,\nr)}{L^*}\right) \label{eq:pre-log_L}
\end{equation}
where $L^*  =  \left \lfloor \frac{1}{2\lambda_D} \right \rfloor$.
\end{theorem}
\begin{proof}
See Section \ref{sec:proof-theorem-1}.
\end{proof}

\begin{remark} 
We derive Theorem \ref{th:pre-log-MIMO} for i.i.d. Gaussian codebooks, which satisfy the average-power constraint \eqref{eq:pw_constraint}. Nevertheless, it can be shown that Theorem \ref{th:pre-log-MIMO} continues to hold when the channel inputs satisfy a peak-power constraint. More specifically, we show in Section \ref{subsec:note-input-distribution} that a sufficient condition on the input distribution with power constraint $\mathsf{E} \left[\left \| \bar \xrv \right \|^2 \right] \leq \nt$ for achieving the pre-log is that its  probability density function (p.d.f.) $p_{\xrv} (\bar \xv)$ satisfies
\begin{equation}
p_{\xrv} (\bar \xv)  \leq \frac{K}{\pi^\nt} e^{-\| \bar \xv \|^2},   \quad \bar \xv \in \field^\nt \label{eq:pdf-constraint}
\end{equation} 
for some $K$ satisfying 
\begin{equation}
\lim_{\SNR \to \infty} \: \frac{\log K}{\log \SNR} = 0. \label{eq:peak-density-cons}
\end{equation}
The condition \eqref{eq:pdf-constraint} is satisfied, for example, by truncated Gaussian inputs, for which the $\nt$ elements in $\bar \xrv$ are independent and identically distributed and 
\begin{align}
  p_{\xrv} (\bar \xv) = \frac{1}{\hat K\pi^\nt} e^{-|\bar \xv|^2},   \quad \bar \xv \in \left \{\bar \xv \in \field^\nt: |\bar x(t)| \leq 1, \: 1 \leq t \leq \nt \right \} 
\end{align}
with
\begin{equation}
\hat K = \left({\int_{|\bar x | \leq 1} \frac{1}{\pi} e^{-|\bar x|^2} d \bar x} \right)^\nt.
\end{equation}
\end{remark}

If $1/(2\lambda_D)$ is an integer, then \eqref{eq:pre-log_L} becomes
\begin{equation}
\Pi_{R^*} \geq \min(\nt, \nr) \big(1 - \min(\nt,\nr)\,2\lambda_D\big).\label{eq:pre-log_lambda}
\end{equation}
Thus, in this case nearest neighbor decoding together with pilot-aided channel estimation achieves the capacity pre-log of MISO fading channels \eqref{eq:MISO-capacity-prelog-BL} as well as the lower bound on the capacity pre-log of MIMO fading channels \eqref{eq:MIMO-capacity-prelog-BL}.

Suppose that both the transmitter and the receiver use the same number of antennas, namely $\nt' \triangleq \nr' \triangleq \min(\nt,\nr)$. Then, as the codeword length tends to infinity, we have from \eqref{eq:total_length}--\eqref{eq:number-for-guard} that the fraction of time consumed for the transmission of pilots is given by
\begin{equation}
\lim_{n \to \infty} \frac{n_{\rm p}}{n'} = \lim_{n \to \infty} \frac{\left(\frac{n}{L - \nt} + 1  + 2 (T-1) \right) \nt'}{\left(\frac{n}{L - \nt'} + 1  + 2 (T-1) \right) \nt' + n + 2(L -\nt')(T-1)} = \frac{\nt'}{L}.
\end{equation}
Consequently, from the achievable pre-log \eqref{eq:pre-log_L}, namely 
\begin{equation}
\Pi_{R^*} \geq \nt' \left(1 - \frac{\nt'}{L}  \right), \qquad L \leq \frac{1}{2\lambda_D}, \label{eq:pre-log-ntprime-min-nt-nr}
\end{equation}
we observe that the loss compared to the capacity pre-log of the coherent fading channel $\nt' = \min(\nt,\nr)$ is given by the fraction of time used for the transmission of pilots. From this we infer that the nearest neighbor decoder in combination with the channel estimator described in Section \ref{sec:model} is optimal at high SNR in the sense that it achieves the capacity pre-log of the coherent fading channel. This further implies that the achievable pre-log in Theorem \ref{th:pre-log-MIMO} is the best pre-log that can be achieved by any scheme employing $\nt'$ pilot vectors.

To achieve the pre-log in Theorem \ref{th:pre-log-MIMO}, we assume that the training period $L$ satisfies $L \leq \frac{1}{2\lambda_D}$, in which case the variance of the interpolation error  \eqref{eq:fadingestimate_error}, namely
\begin{equation}
 \epsilon^2 = 1 - \int^{1/2}_{-1/2} \frac{\SNR \left[ f_H (\lambda) \right]^2}{\SNR f_H(\lambda) + L \nt } d \lambda \approx \frac{2\lambda_D L \nt}{\SNR},
\end{equation}
vanishes as the inverse of the SNR. The achievable pre-log is then maximized by maximizing $L \leq \frac{1}{2\lambda_D}$. Note that as a criterion of ``perfect side information''  for nearest neighbor decoding in fading channels, Lapidoth and Shamai \cite{IEEE:lapidoth:fadingchannels_howperfect} suggested that the variance of the fading estimation error should be negligible compared to the reciprocal of the SNR. Using the linear interpolator \eqref{eq:LMMSEestimation}, we obtain an estimation error with variance decaying as the reciprocal of the SNR provided that $L \leq \frac{1}{2\lambda_D}$. Thus, the condition $L \leq \frac{1}{2\lambda_D}$ can be viewed as a sufficient condition for obtaining ``nearly perfect side information'' in the sense that the variance of the interpolation error is of the same order as the reciprocal of the SNR.

Of course, one could increase the training period $L$ beyond $\frac{1}{2\lambda_D}$. Indeed, by increasing $L$, we could reduce the rate loss due to the transmission of pilots as indicated in \eqref{eq:pre-log-ntprime-min-nt-nr} at the cost of obtaining a larger fading estimation error, which in turn may reduce the reliability of the nearest neighbor decoder. To understand this trade-off better, we shall analyze the achievable pre-log when $L > \frac{1}{2\lambda_D}$. Note that for $L > \frac{1}{2\lambda_D}$, the variance of the interpolation error follows from \eqref{eq:interpolation-error-variance-T-infinity-general} 
\begin{align}
 \epsilon^2_\ell (t) & =  1 - \int^{1/2}_{-1/2} \frac{\SNR \left|f_{L,\ell - t + 1} (\lambda) \right|^2}{\SNR f_{L,0} (\lambda) + \nt} d \lambda \\
& =  \int^{1/2}_{-1/2} \frac{\nt f_{L,0} (\lambda)}{\SNR f_{L,0} (\lambda) + \nt} d\lambda + \int^{1/2}_{-1/2} \frac{\SNR \left( \left[ f_{L, 0} (\lambda) \right]^2 - \left|f_{L, \ell-t+1} (\lambda) \right|^2 \right)}{\SNR f_{L,0} (\lambda) + \nt} d\lambda.   \label{eq:epsilon-sq-ell-t-L-geq-inv-2lambdaD}
\end{align}
The former integral 
\begin{equation}
 \int^{1/2}_{-1/2} \frac{\nt f_{L,0} (\lambda)}{\SNR f_{L,0} (\lambda) + \nt} d\lambda \approx \frac{\nt}{\SNR}
\end{equation}
vanishes as the SNR tends to infinity. However, we prove in Appendix \ref{sec:proof-lemma:interpolator-property-T-infinity} that as the SNR tends to infinity, the latter integral 
\begin{equation}
\int^{1/2}_{-1/2} \frac{\SNR \left( \left[ f_{L, 0} (\lambda) \right]^2 - \left|f_{L, \ell-t+1} (\lambda) \right|^2 \right)}{\SNR f_{L,0} (\lambda) + \nt} d\lambda
\end{equation}
is bounded away from zero. This implies that the interpolation error \eqref{eq:epsilon-sq-ell-t-L-geq-inv-2lambdaD} does not vanish as the SNR tends to infinity, and the pre-log achievable with the scheme described in Section \ref{sec:model} is zero. It thus follows that the condition $L \leq \frac{1}{2\lambda_D}$ is necessary in order to achieve a positive pre-log.

Comparing \eqref{eq:pre-log_L} and \eqref{eq:MIMO-capacity-prelog-BL} with the capacity pre-log $\min(\nt,\nr)$ for coherent fading channels \cite{Bell_foschini_layered_space-time,telatar_multiantenna_Gaussian}, we observe that, for a fading process of bandwidth $\lambda_D$, the penalty for not knowing the fading coefficients is roughly $(\min(\nt,\nr))^2 \cdot 2\lambda_D$. Consequently, the lower bound \eqref{eq:pre-log_L} does not grow linearly with $\min(\nt,\nr)$, but it is a quadratic function of $\min(\nt,\nr)$ that achieves its maximum at
\begin{equation}
\min(\nt,\nr) = \frac{L^*}{2}.
\end{equation}
This gives rise to the lower bound
\begin{equation}
\Pi_{R^*} \geq \frac{L^*}{4}
\end{equation}
which cannot be larger than $1/(8\lambda_D)$. The same holds for the lower bound \eqref{eq:etaMIMO}.

\section{Fading Multiple-Access Channels }
\label{sec:mac}

In this section, we extend the use of nearest neighbor decoding with pilot-aided channel estimation to the fading MAC. We are interested in the achievable pre-log region that can be achieved with this scheme. 

\begin{figure}[t]
\begin{center}
\includegraphics[width=0.9\textwidth]{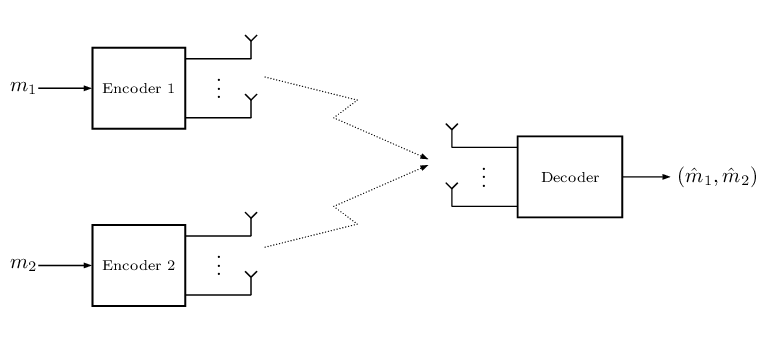}
\vspace*{-2mm}
\end{center}
\caption{The two-user MIMO fading MAC diagram.}
\label{fig:mac-system-model}
\vspace*{-5mm}
\end{figure}


We consider a two-user MIMO fading MAC, where two terminals wish to communicate with a third one, and where the channels between the terminals are MIMO fading channels. Extension to more than two users is straightforward. The first user has $\ntone$ antennas, the second user has $\nttwo$ antennas and the receiver has $\nr$ antennas. The channel model is depicted in Fig.~\ref{fig:mac-system-model}. The channel output at time instant $k \in \integ$ is a complex-valued $\nr$-dimensional random vector given by 
\begin{equation}
\label{eq:channel}
 \yrv_k = \sqrt{{\sf SNR}}\, \HRM_{1,k} \xv_{1,k} + \sqrt{{\sf SNR}}\, \HRM_{2,k} \xv_{2,k} + \zrv_k.
\end{equation}
Here  $\xv_{s,k} \in \field^{\nts}$ denotes the time-$k$ channel input vector corresponding to User $s$, $s=1,2$; $\HRM_{s,k}$ denotes the $(\nr \times \nts)$-dimensional fading matrix at time $k$ corresponding to User $s$, $s=1,2$; $\SNR$ denotes the average SNR for each transmit antenna; and $\zrv_k $ denotes the $\nr$-variate additive noise vector at time $k$. The fading processes $\{\HRM_{s,k}, k \in \integ\}$, $s=1,2$ are independent of each other and of the noise process $\{\zrv_k, k \in \integ \}$, and follow the same setup as the one used in the point-to-point channel (Section \ref{sec:model}).

Both users transmit codewords and pilot symbols over the channel \eqref{eq:channel}. To transmit the message $m_s \in \{1,\dotsc,\lfloor e^{nR_s} \rfloor \}$, $s=1,2$, (where $m_1$ and $m_2$ are drawn independently) each user's encoder selects a codeword of length $n$ from a codebook $\mathcal{C}_s$, where $\mathcal{C}_s$, $s=1,2$ are drawn i.i.d.\ from an $\nts$-variate, zero-mean, complex-Gaussian distribution of covariance matrix ${\sf I}_{\nts}$. Similar to the single-user case, orthogonal pilot vectors are used. The pilot vector ${\bm p}_{s,t} \in \field^{\nts}$, $s=1,2$, $t=1,\dotsc,\nts$ used to estimate the fading coefficients from transmit antenna $t$ of User $s$ is given by $p_{s,t} (t) = 1$ and $p_{s,t} (t') = 0$ for $t' \neq t$. For example, the first pilot vector of User $s$ is given by $\trans{(1,0,\dotsc,0)}$. To estimate the fading matrices $\HRM_{1,k}$ and $\HRM_{2,k}$, each training period requires transmission of $(\ntone + \nttwo)$ pilot vectors ${\bm p}_{1,1},\dotsc, {\bm p}_{1,\ntone}, {\bm p}_{2,1},\dotsc, {\bm p}_{2,\nttwo}$.

Assuming transmission from both users is synchronized, the transmission scheme extends the point-to-point setup in Section \ref{sec:model} to the two-user MAC setup as illustrated in Fig. \ref{fig:pilot_data_illustration_JTD}. Every $L$ time instants (for some $L \geq \ntone + \nttwo,~ L\in\naturals$), User 1 first transmits the $\ntone$ pilot vectors ${\bm p}_{1,1},\dotsc, {\bm p}_{1,\ntone}$. Once the transmission of the $\ntone$ pilot vectors ends, User 2 transmits its $\nttwo$ pilot vectors ${\bm p}_{2,1},\dotsc, {\bm p}_{2,\nttwo}$.  The codewords for both users are then split up into blocks of $(L- \ntone -\nttwo)$ data vectors, which are transmitted simultaneously after the $(\ntone + \nttwo)$ pilot vectors. The process of transmitting $(L- \ntone -\nttwo)$ data vectors and $(\ntone + \nttwo)$ pilot vectors continues until all $n$ data symbols are completed. Herein we assume that $n$ is an integer multiple of $(L - \ntone - \nttwo)$.\footnote{As in the point-to-point setup, this assumption is not critical in terms of rate, cf. Footnote 1 on page 5.} Prior to transmitting the first data block, and after transmitting the last data block, a guard period of $L (T -1)$ time instants (for some $T\in\naturals$) is introduced for the purpose of channel estimation, where we transmit every $L$ time instants the $(\ntone + \nttwo)$ pilot vectors but we do not transmit data vectors in between. Note that codewords from both users are \emph{jointly} transmitted at the same time instants whereas pilots from both users do not interfere and are \emph{separately} transmitted at different time instants. The total block-length of this transmission scheme (comprising data vectors, pilot vectors and guard period) is given by
\begin{equation}
 n' = n_{\rm p} + n + n_{\rm g} \label{eq:total_length_mac}
\end{equation}
where $n_{\rm p}$ and $n_{\rm g}$ are
\begin{align}
 n_{\rm p} &= \left(\frac{n}{L - \ntone - \nttwo} + 1  + 2 (T-1) \right) (\ntone + \nttwo), \IEEEeqnarraynumspace\\
 n_{\rm g} &= 2(L- \ntone - \nttwo)(T-1).
\end{align}

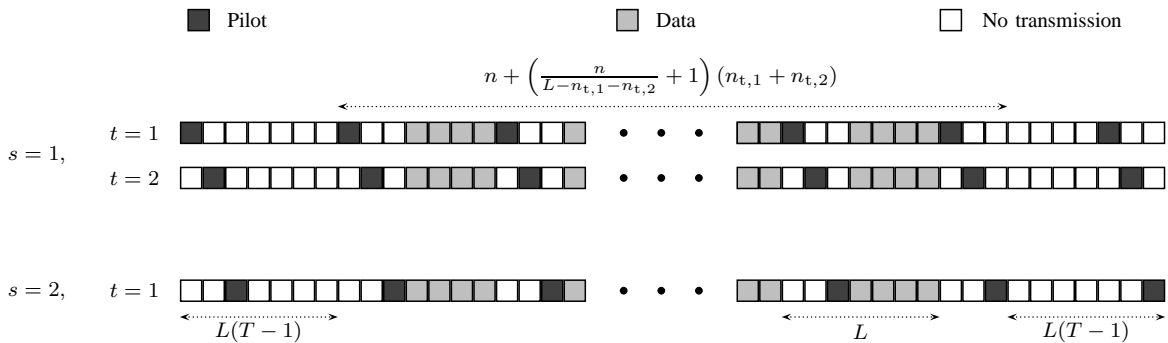
\begin{figure*}[t]
\begin{center}
\begin{pspicture}(-3.75,-2.65)(12,1.8)

\psframe[linewidth=0.02,fillstyle=solid,fillcolor=darkgray](-2,1.5)(-1.7,1.8)
\rput(-1.2,1.65){\scriptsize Pilot}

\psframe[linewidth=0.02,fillstyle=solid,fillcolor=lightgray](3.7,1.5)(4,1.8)
\rput(4.5,1.65){\scriptsize Data}

\psframe[linewidth=0.02](8.0,1.5)(8.3,1.8)
\rput(9.5,1.65){\scriptsize No transmission}

\rput(4.3,0.9){\scriptsize $n + \left(\frac{n}{L - \ntone - \nttwo} + 1 \right) (\ntone + \nttwo)$}
\psline[linewidth=0.02,linestyle=dotted,dotsep=1pt]{<->}(0,0.5)(8.9,0.5)

\rput(6.95,-2.5){\scriptsize $L$}
\psline[linewidth=0.02,linestyle=dotted,dotsep=1pt]{<->}(5.9,-2.3)(8,-2.3)

\rput(9.95,-2.5){\scriptsize $L(T-1)$}
\psline[linewidth=0.02,linestyle=dotted,dotsep=1pt]{<->}(8.9,-2.3)(11,-2.3)

\rput(-1.05,-2.5){\scriptsize $L(T-1)$}
\psline[linewidth=0.02,linestyle=dotted,dotsep=1pt]{<->}(-2.1,-2.3)(0,-2.3)

\rput(-4,-0.15){\scriptsize $s=1,$}

{
\rput(-2.7,0.15){\scriptsize $t=1$}

\psframe[linewidth=0.02,fillstyle=solid,fillcolor=darkgray](-2.1,0)(-1.8,0.3)
\psframe[linewidth=0.02](-1.8,0)(-1.5,0.3)
\psframe[linewidth=0.02](-1.5,0)(-1.2,0.3)
\psframe[linewidth=0.02](-1.2,0)(-0.9,0.3)
\psframe[linewidth=0.02](-0.9,0)(-0.6,0.3)
\psframe[linewidth=0.02](-0.6,0)(-0.3,0.3)
\psframe[linewidth=0.02](-0.3,0)(0,0.3)

\psframe[linewidth=0.02,fillstyle=solid,fillcolor=darkgray](0,0)(0.3,0.3)
\psframe[linewidth=0.02](0.3,0)(0.6,0.3) 
\psframe[linewidth=0.02](0.6,0)(0.9,0.3)
\psframe[linewidth=0.02,fillstyle=solid,fillcolor=lightgray](0.9,0)(1.2,0.3)
\psframe[linewidth=0.02,fillstyle=solid,fillcolor=lightgray](1.2,0)(1.5,0.3)
\psframe[linewidth=0.02,fillstyle=solid,fillcolor=lightgray](1.5,0)(1.8,0.3)
\psframe[linewidth=0.02,fillstyle=solid,fillcolor=lightgray](1.8,0)(2.1,0.3)

\psframe[linewidth=0.02,fillstyle=solid,fillcolor=darkgray](2.1,0)(2.4,0.3)
\psframe[linewidth=0.02](2.4,0)(2.7,0.3) 
\psframe[linewidth=0.02](2.7,0)(3,0.3)
\psframe[linewidth=0.02,fillstyle=solid,fillcolor=lightgray](3,0)(3.3,0.3)

\pscircle[fillstyle=solid,fillcolor=black,linewidth=0.0875,linecolor=white](3.8,0.15){0.14}
\pscircle[fillstyle=solid,fillcolor=black,linewidth=0.0875,linecolor=white](4.3,0.15){0.14}
\pscircle[fillstyle=solid,fillcolor=black,linewidth=0.0875,linecolor=white](4.8,0.15){0.14}

\psframe[linewidth=0.02](5.3,0)(8.6,0.3)
\psframe[linewidth=0.02,fillstyle=solid,fillcolor=lightgray](5.3,0)(5.6,0.3)
\psframe[linewidth=0.02,fillstyle=solid,fillcolor=lightgray](5.6,0)(5.9,0.3)
\psframe[linewidth=0.02,fillstyle=solid,fillcolor=darkgray](5.9,0)(6.2,0.3)
\psframe[linewidth=0.02](6.2,0)(6.5,0.3)
\psframe[linewidth=0.02](6.5,0)(6.8,0.3)
\psframe[linewidth=0.02,fillstyle=solid,fillcolor=lightgray](6.8,0)(7.1,0.3)
\psframe[linewidth=0.02,fillstyle=solid,fillcolor=lightgray](7.1,0)(7.4,0.3)
\psframe[linewidth=0.02,fillstyle=solid,fillcolor=lightgray](7.4,0)(7.7,0.3)
\psframe[linewidth=0.02,fillstyle=solid,fillcolor=lightgray](7.7,0)(8,0.3)

\psframe[linewidth=0.02,fillstyle=solid,fillcolor=darkgray](8,0)(8.3,0.3)
\psframe[linewidth=0.02](8.3,0)(8.6,0.3)
\psframe[linewidth=0.02](8.6,0)(8.9,0.3)
\psframe[linewidth=0.02](8.9,0)(9.2,0.3)
\psframe[linewidth=0.02](9.2,0)(9.5,0.3)
\psframe[linewidth=0.02](9.5,0)(9.8,0.3)
\psframe[linewidth=0.02](9.8,0)(10.1,0.3)

\psframe[linewidth=0.02,fillstyle=solid,fillcolor=darkgray](10.1,0)(10.4,0.3)
\psframe[linewidth=0.02](10.4,0)(10.7,0.3)
\psframe[linewidth=0.02](10.7,0)(11,0.3)

}

{
\rput(-2.7,-0.45){\scriptsize $t=2$}

\psframe[linewidth=0.02](-2.1,-0.6)(-1.8,-0.3)
\psframe[linewidth=0.02,fillstyle=solid,fillcolor=darkgray](-1.5,-0.6)(-1.8,-0.3)
\psframe[linewidth=0.02](-1.5,-0.6)(-1.2,-0.3) 
\psframe[linewidth=0.02](-1.2,-0.6)(-0.9,-0.3)
\psframe[linewidth=0.02](-0.9,-0.6)(-0.6,-0.3) 
\psframe[linewidth=0.02](-0.6,-0.6)(-0.3,-0.3)
\psframe[linewidth=0.02](-0.3,-0.6)(0,-0.3) 

\psframe[linewidth=0.02](0,-0.6)(0.3,-0.3)
\psframe[linewidth=0.02,fillstyle=solid,fillcolor=darkgray](0.3,-0.6)(0.6,-0.3)
\psframe[linewidth=0.02](0.6,-0.6)(0.9,-0.3)
\psframe[linewidth=0.02,fillstyle=solid,fillcolor=lightgray](0.9,-0.6)(1.2,-0.3)
\psframe[linewidth=0.02,fillstyle=solid,fillcolor=lightgray](1.2,-0.6)(1.5,-0.3)
\psframe[linewidth=0.02,fillstyle=solid,fillcolor=lightgray](1.5,-0.6)(1.8,-0.3)
\psframe[linewidth=0.02,fillstyle=solid,fillcolor=lightgray](1.8,-0.6)(2.1,-0.3)

\psframe[linewidth=0.02](2.1,-0.6)(2.4,-0.3)
\psframe[linewidth=0.02,fillstyle=solid,fillcolor=darkgray](2.4,-0.6)(2.7,-0.3)
\psframe[linewidth=0.02](2.7,-0.6)(3,-0.3)
\psframe[linewidth=0.02,fillstyle=solid,fillcolor=lightgray](3,-0.6)(3.3,-0.3)

\pscircle[fillstyle=solid,fillcolor=black,linewidth=0.0875,linecolor=white](3.8,-0.45){0.14}
\pscircle[fillstyle=solid,fillcolor=black,linewidth=0.0875,linecolor=white](4.3,-0.45){0.14}
\pscircle[fillstyle=solid,fillcolor=black,linewidth=0.0875,linecolor=white](4.8,-0.45){0.14}

\psframe[linewidth=0.02,fillstyle=solid,fillcolor=lightgray](5.3,-0.6)(5.6,-0.3)
\psframe[linewidth=0.02,fillstyle=solid,fillcolor=lightgray](5.6,-0.6)(5.9,-0.3)
\psframe[linewidth=0.02](5.9,-0.6)(6.2,-0.3)
\psframe[linewidth=0.02,fillstyle=solid,fillcolor=darkgray](6.2,-0.6)(6.5,-0.3)

\psframe[linewidth=0.02](6.5,-0.6)(6.8,-0.3)
\psframe[linewidth=0.02,fillstyle=solid,fillcolor=lightgray](6.8,-0.6)(7.1,-0.3)
\psframe[linewidth=0.02,fillstyle=solid,fillcolor=lightgray](7.1,-0.6)(7.4,-0.3)
\psframe[linewidth=0.02,fillstyle=solid,fillcolor=lightgray](7.4,-0.6)(7.7,-0.3)
\psframe[linewidth=0.02,fillstyle=solid,fillcolor=lightgray](7.7,-0.6)(8,-0.3)

\psframe[linewidth=0.02](8,-0.6)(8.3,-0.3)
\psframe[linewidth=0.02,fillstyle=solid,fillcolor=darkgray](8.3,-0.6)(8.6,-0.3)
\psframe[linewidth=0.02](8.6,-0.6)(8.9,-0.3)
\psframe[linewidth=0.02](8.9,-0.6)(9.2,-0.3)
\psframe[linewidth=0.02](9.2,-0.6)(9.5,-0.3)
\psframe[linewidth=0.02](9.5,-0.6)(9.8,-0.3)
\psframe[linewidth=0.02](9.8,-0.6)(10.1,-0.3)

\psframe[linewidth=0.02](10.1,-0.6)(10.4,-0.3)
\psframe[linewidth=0.02,fillstyle=solid,fillcolor=darkgray](10.7,-0.6)(10.4,-0.3)
\psframe[linewidth=0.02](11,-0.6)(10.7,-0.3)

}

{
\rput(-4,-1.95){\scriptsize $s=2,$}

\rput(-2.7,-1.95){\scriptsize $t=1$}

\psframe[linewidth=0.02](-2.1,-2.1)(-1.8,-1.8)
\psframe[linewidth=0.02](-1.5,-2.1)(-1.8,-1.8)
\psframe[linewidth=0.02,fillstyle=solid,fillcolor=darkgray](-1.5,-2.1)(-1.2,-1.8) 
\psframe[linewidth=0.02](-1.2,-2.1)(-0.9,-1.8)
\psframe[linewidth=0.02](-0.9,-2.1)(-0.6,-1.8) 
\psframe[linewidth=0.02](-0.6,-2.1)(-0.3,-1.8)
\psframe[linewidth=0.02](-0.3,-2.1)(0,-1.8) 

\psframe[linewidth=0.02](0,-2.1)(0.3,-1.8)
\psframe[linewidth=0.02](0.3,-2.1)(0.6,-1.8)
\psframe[linewidth=0.02,fillstyle=solid,fillcolor=darkgray](0.6,-2.1)(0.9,-1.8)
\psframe[linewidth=0.02,fillstyle=solid,fillcolor=lightgray](0.9,-2.1)(1.2,-1.8)
\psframe[linewidth=0.02,fillstyle=solid,fillcolor=lightgray](1.2,-2.1)(1.5,-1.8)
\psframe[linewidth=0.02,fillstyle=solid,fillcolor=lightgray](1.5,-2.1)(1.8,-1.8)
\psframe[linewidth=0.02,fillstyle=solid,fillcolor=lightgray](1.8,-2.1)(2.1,-1.8)

\psframe[linewidth=0.02](2.1,-2.1)(2.4,-1.8)
\psframe[linewidth=0.02](2.4,-2.1)(2.7,-1.8)
\psframe[linewidth=0.02,fillstyle=solid,fillcolor=darkgray](2.7,-2.1)(3,-1.8)
\psframe[linewidth=0.02,fillstyle=solid,fillcolor=lightgray](3,-2.1)(3.3,-1.8)

\pscircle[fillstyle=solid,fillcolor=black,linewidth=0.0875,linecolor=white](3.8,-1.95){0.14}
\pscircle[fillstyle=solid,fillcolor=black,linewidth=0.0875,linecolor=white](4.3,-1.95){0.14}
\pscircle[fillstyle=solid,fillcolor=black,linewidth=0.0875,linecolor=white](4.8,-1.95){0.14}

\psframe[linewidth=0.02,fillstyle=solid,fillcolor=lightgray](5.3,-2.1)(5.6,-1.8)
\psframe[linewidth=0.02,fillstyle=solid,fillcolor=lightgray](5.6,-2.1)(5.9,-1.8)
\psframe[linewidth=0.02](5.9,-2.1)(6.2,-1.8)
\psframe[linewidth=0.02](6.2,-2.1)(6.5,-1.8)

\psframe[linewidth=0.02,fillstyle=solid,fillcolor=darkgray](6.5,-2.1)(6.8,-1.8)
\psframe[linewidth=0.02,fillstyle=solid,fillcolor=lightgray](6.8,-2.1)(7.1,-1.8)
\psframe[linewidth=0.02,fillstyle=solid,fillcolor=lightgray](7.1,-2.1)(7.4,-1.8)
\psframe[linewidth=0.02,fillstyle=solid,fillcolor=lightgray](7.4,-2.1)(7.7,-1.8)
\psframe[linewidth=0.02,fillstyle=solid,fillcolor=lightgray](7.7,-2.1)(8,-1.8)

\psframe[linewidth=0.02](8,-2.1)(8.3,-1.8)
\psframe[linewidth=0.02](8.3,-2.1)(8.6,-1.8)
\psframe[linewidth=0.02,fillstyle=solid,fillcolor=darkgray](8.6,-2.1)(8.9,-1.8)
\psframe[linewidth=0.02](8.9,-2.1)(9.2,-1.8)
\psframe[linewidth=0.02](9.2,-2.1)(9.5,-1.8)
\psframe[linewidth=0.02](9.5,-2.1)(9.8,-1.8)
\psframe[linewidth=0.02](9.8,-2.1)(10.1,-1.8)

\psframe[linewidth=0.02](10.1,-2.1)(10.4,-1.8)
\psframe[linewidth=0.02](10.7,-2.1)(10.4,-1.8)
\psframe[linewidth=0.02,fillstyle=solid,fillcolor=darkgray](11,-2.1)(10.7,-1.8)

}

\end{pspicture}
\end{center}
\caption{Structure of joint-transmission scheme, $\ntone = 2$, $\nttwo = 1$, $L = 7$ and $T=2$.}
\label{fig:pilot_data_illustration_JTD}
\end{figure*}
Similarly to the single-user case, the receiver guesses which messages have been transmitted using a two-part decoder that consists a channel estimator and  a data detector. The channel estimator first obtains matrix-valued fading estimates $\{ \hat \HRM_{s,k}^{(T)}, k \in \mathcal{D} \}$, $s=1,2$ from the received pilots $\yrv_{k'}$, $k' \in \mathcal{P}$ using the same linear interpolator as \eqref{eq:LMMSEestimation}. From the received codeword $\{\yv_k, k \in \mathcal{D} \}$ and the channel-estimate matrices $\{\hat \HM_{s,k}^{(T)}, k \in \mathcal{D} \}$, $s=1,2$ (which are the realizations of $\{ \hat \HRM_{s,k}^{(T)}, k \in \mathcal{D} \}$, $s=1,2$), the decoder chooses the pair of messages $(\hat m_1, \hat m_2)$ that minimizes the distance metric
\begin{equation}
(\hat m_1, \hat m_2)  =  \arg \min_{(m_1, m_2)} D(m_1,m_2) 
\end{equation}
where 
\begin{equation}
D(m_1,m_2)  \triangleq  \sum_{k \in \mathcal{D}^{(n')} } \left \|\yv_k - \sqrt{\SNR}\, \hat \HM^{(T)}_{1,k} \xv_{1,k} {(m_1)} - \sqrt{\SNR}\: \hat \HM^{(T)}_{2,k}\xv_{2,k} {(m_2)} \right \|^2
 \label{eq:decoding-rule-JTD}
\end{equation}
and where $\mathcal{D}^{(n')}$ is defined in the same way as \eqref{eq:Dnprime-defined}. In the following, we will refer to the above communication scheme as the \emph{joint-transmission scheme}.

\begin{figure*}[t]
\begin{center}
\begin{pspicture}(-4,-3)(12,1.2)

\psframe[linewidth=0.02,fillstyle=solid,fillcolor=darkgray](-2,0.8)(-1.7,1.1)
\rput(-1.2,0.95){\scriptsize Pilot}

\psframe[linewidth=0.02,fillstyle=solid,fillcolor=lightgray](3.7,0.8)(4,1.1)
\rput(4.5,0.95){\scriptsize Data}

\psframe[linewidth=0.02](8.0,0.8)(8.3,1.1)
\rput(9.5,0.95){\scriptsize No transmission}


\rput(-1.5,-1){\scriptsize $L(T-1)$}
\psline[linewidth=0.02,linestyle=dotted,dotsep=1pt]{<->}(-2.1,-0.8)(-0.9,-0.8)

\rput(3,-1){\scriptsize $L$}
\psline[linewidth=0.02,linestyle=dotted,dotsep=1pt]{<->}(2.4,-0.8)(3.6,-0.8)

\rput(4.8,-1){\scriptsize $L(T-1)$}
\psline[linewidth=0.02,linestyle=dotted,dotsep=1pt]{<->}(4.2,-0.8)(5.4,-0.8)

\rput(6,-1.4){\scriptsize $L(T-1)$}
\psline[linewidth=0.02,linestyle=dotted,dotsep=1pt]{<->}(5.4,-1.6)(6.6,-1.6)

\rput(7.2,-2.5){\scriptsize $L$}
\psline[linewidth=0.02,linestyle=dotted,dotsep=1pt]{<->}(6.6,-2.3)(7.8,-2.3)

\rput(10.9,-1.4){\scriptsize $L(T-1)$}
\psline[linewidth=0.02,linestyle=dotted,dotsep=1pt]{<->}(10.3,-1.6)(11.5,-1.6)

\psline[linewidth=0.02,linestyle=dotted,dotsep=1pt](5.4,0.6)(5.4,-3)

\rput(1.65,-3){\scriptsize $\beta n'$} 
\psline[linewidth=0.02,linestyle=dotted,dotsep=1pt]{<->}(-2.1,-2.8)(5.4,-2.8)

\rput(8.45,-3){\scriptsize $(1 - \beta) n'$} 
\psline[linewidth=0.02,linestyle=dotted,dotsep=1pt]{<->}(5.4,-2.8)(11.5,-2.8)

\rput(-4,-0.15){\scriptsize $s=1,$}

{
\rput(-2.7,0.15){\scriptsize $t=1$}

\psframe[linewidth=0.02,fillstyle=solid,fillcolor=darkgray](-2.1,0)(-1.8,0.3)
\psframe[linewidth=0.02](-1.8,0)(-1.5,0.3)
\psframe[linewidth=0.02](-1.5,0)(-1.2,0.3)
\psframe[linewidth=0.02](-1.2,0)(-0.9,0.3)
\psframe[linewidth=0.02,fillstyle=solid,fillcolor=darkgray](-0.9,0)(-0.6,0.3)
\psframe[linewidth=0.02](-0.6,0)(-0.3,0.3)
\psframe[linewidth=0.02,fillstyle=solid,fillcolor=lightgray](-0.3,0)(0,0.3)
\psframe[linewidth=0.02,fillstyle=solid,fillcolor=lightgray](0,0)(0.3,0.3)
\psframe[linewidth=0.02,fillstyle=solid,fillcolor=darkgray](0.3,0)(0.6,0.3)
\psframe[linewidth=0.02](0.6,0)(0.9,0.3)
\psframe[linewidth=0.02,fillstyle=solid,fillcolor=lightgray](0.9,0)(1.2,0.3)
\psframe[linewidth=0.02,fillstyle=solid,fillcolor=lightgray](1.2,0)(1.5,0.3)

\pscircle[fillstyle=solid,fillcolor=black,linewidth=0.0875,linecolor=white](1.75,0.15){0.14}
\pscircle[fillstyle=solid,fillcolor=black,linewidth=0.0875,linecolor=white](2,0.15){0.14}
\pscircle[fillstyle=solid,fillcolor=black,linewidth=0.0875,linecolor=white](2.25,0.15){0.14}

\psframe[linewidth=0.02,fillstyle=solid,fillcolor=darkgray](2.4,0)(2.7,0.3) 
\psframe[linewidth=0.02](2.7,0)(3,0.3)
\psframe[linewidth=0.02,fillstyle=solid,fillcolor=lightgray](3,0)(3.3,0.3)
\psframe[linewidth=0.02,fillstyle=solid,fillcolor=lightgray](3.3,0)(3.6,0.3) 
\psframe[linewidth=0.02,fillstyle=solid,fillcolor=darkgray](3.6,0)(3.9,0.3)
\psframe[linewidth=0.02](3.9,0)(4.2,0.3)
\psframe[linewidth=0.02](4.2,0)(4.5,0.3)
\psframe[linewidth=0.02](4.5,0)(4.8,0.3)
\psframe[linewidth=0.02,fillstyle=solid,fillcolor=darkgray](4.8,0)(5.1,0.3)
\psframe[linewidth=0.02](5.1,0)(5.4,0.3)
\psframe[linewidth=0.02](5.4,0)(5.7,0.3)
\psframe[linewidth=0.02](5.7,0)(6,0.3)
\psframe[linewidth=0.02](6,0)(6.3,0.3)
\psframe[linewidth=0.02](6.3,0)(6.6,0.3)
\psframe[linewidth=0.02](6.6,0)(6.9,0.3)
\psframe[linewidth=0.02](6.9,0)(7.2,0.3)
\psframe[linewidth=0.02](7.2,0)(7.5,0.3)
\psframe[linewidth=0.02](7.5,0)(7.8,0.3)
\psframe[linewidth=0.02](7.8,0)(8.1,0.3)
\psframe[linewidth=0.02](8.1,0)(8.4,0.3)
\psframe[linewidth=0.02](8.4,0)(8.7,0.3)
\psframe[linewidth=0.02](8.7,0)(9,0.3)

\pscircle[fillstyle=solid,fillcolor=black,linewidth=0.0875,linecolor=white](9.25,0.15){0.14}
\pscircle[fillstyle=solid,fillcolor=black,linewidth=0.0875,linecolor=white](9.5,0.15){0.14}
\pscircle[fillstyle=solid,fillcolor=black,linewidth=0.0875,linecolor=white](9.75,0.15){0.14}

\psframe[linewidth=0.02](10.0,0)(10.3,0.3)
\psframe[linewidth=0.02](10.3,0)(10.6,0.3)
\psframe[linewidth=0.02](10.6,0)(10.9,0.3)
\psframe[linewidth=0.02](10.9,0)(11.2,0.3)
\psframe[linewidth=0.02](11.2,0)(11.5,0.3)

}

{
\rput(-2.7,-0.45){\scriptsize $t=2$}

\psframe[linewidth=0.02](-2.1,-0.6)(-1.8,-0.3)
\psframe[linewidth=0.02,fillstyle=solid,fillcolor=darkgray](-1.5,-0.6)(-1.8,-0.3)
\psframe[linewidth=0.02](-1.5,-0.6)(-1.2,-0.3) 
\psframe[linewidth=0.02](-1.2,-0.6)(-0.9,-0.3)
\psframe[linewidth=0.02](-0.9,-0.6)(-0.6,-0.3) 
\psframe[linewidth=0.02,fillstyle=solid,fillcolor=darkgray](-0.6,-0.6)(-0.3,-0.3)
\psframe[linewidth=0.02,fillstyle=solid,fillcolor=lightgray](-0.3,-0.6)(0,-0.3) 
\psframe[linewidth=0.02,fillstyle=solid,fillcolor=lightgray](0,-0.6)(0.3,-0.3)
\psframe[linewidth=0.02](0.3,-0.6)(0.6,-0.3)
\psframe[linewidth=0.02,fillstyle=solid,fillcolor=darkgray](0.6,-0.6)(0.9,-0.3)
\psframe[linewidth=0.02,fillstyle=solid,fillcolor=lightgray](0.9,-0.6)(1.2,-0.3)
\psframe[linewidth=0.02,fillstyle=solid,fillcolor=lightgray](1.2,-0.6)(1.5,-0.3)

\pscircle[fillstyle=solid,fillcolor=black,linewidth=0.0875,linecolor=white](1.75,-0.45){0.14}
\pscircle[fillstyle=solid,fillcolor=black,linewidth=0.0875,linecolor=white](2,-0.45){0.14}
\pscircle[fillstyle=solid,fillcolor=black,linewidth=0.0875,linecolor=white](2.25,-0.45){0.14}

\psframe[linewidth=0.02](2.4,-0.6)(2.7,-0.3)
\psframe[linewidth=0.02,fillstyle=solid,fillcolor=darkgray](2.7,-0.6)(3,-0.3)
\psframe[linewidth=0.02,fillstyle=solid,fillcolor=lightgray](3,-0.6)(3.3,-0.3)
\psframe[linewidth=0.02,fillstyle=solid,fillcolor=lightgray](3.3,-0.6)(3.6,-0.3)
\psframe[linewidth=0.02](3.6,-0.6)(3.9,-0.3)
\psframe[linewidth=0.02,fillstyle=solid,fillcolor=darkgray](3.9,-0.6)(4.2,-0.3)
\psframe[linewidth=0.02](4.2,-0.6)(4.5,-0.3)
\psframe[linewidth=0.02](4.5,-0.6)(4.8,-0.3)
\psframe[linewidth=0.02](4.8,-0.6)(5.1,-0.3)
\psframe[linewidth=0.02,fillstyle=solid,fillcolor=darkgray](5.1,-0.6)(5.4,-0.3)
\psframe[linewidth=0.02](5.4,-0.6)(5.7,-0.3)
\psframe[linewidth=0.02](5.7,-0.6)(6,-0.3)
\psframe[linewidth=0.02](6,-0.6)(6.3,-0.3)
\psframe[linewidth=0.02](6.3,-0.6)(6.6,-0.3)
\psframe[linewidth=0.02](6.6,-0.6)(6.9,-0.3)
\psframe[linewidth=0.02](6.9,-0.6)(7.2,-0.3)
\psframe[linewidth=0.02](7.2,-0.6)(7.5,-0.3)
\psframe[linewidth=0.02](7.5,-0.6)(7.8,-0.3)
\psframe[linewidth=0.02](7.8,-0.6)(8.1,-0.3)
\psframe[linewidth=0.02](8.1,-0.6)(8.4,-0.3)
\psframe[linewidth=0.02](8.4,-0.6)(8.7,-0.3)
\psframe[linewidth=0.02](8.7,-0.6)(9,-0.3)

\pscircle[fillstyle=solid,fillcolor=black,linewidth=0.0875,linecolor=white](9.25,-0.45){0.14}
\pscircle[fillstyle=solid,fillcolor=black,linewidth=0.0875,linecolor=white](9.5,-0.45){0.14}
\pscircle[fillstyle=solid,fillcolor=black,linewidth=0.0875,linecolor=white](9.75,-0.45){0.14}

\psframe[linewidth=0.02](10,-0.6)(10.3,-0.3)
\psframe[linewidth=0.02](10.3,-0.6)(10.6,-0.3)
\psframe[linewidth=0.02](10.6,-0.6)(10.9,-0.3)
\psframe[linewidth=0.02](10.9,-0.6)(11.2,-0.3)
\psframe[linewidth=0.02](11.2,-0.6)(11.5,-0.3)

}

{
\rput(-4,-1.95){\scriptsize $s=2,$}

\rput(-2.7,-1.95){\scriptsize $t=1$}

\psframe[linewidth=0.02](-2.1,-2.1)(-1.8,-1.8)
\psframe[linewidth=0.02](-1.5,-2.1)(-1.8,-1.8)
\psframe[linewidth=0.02](-1.5,-2.1)(-1.2,-1.8) 
\psframe[linewidth=0.02](-1.2,-2.1)(-0.9,-1.8)
\psframe[linewidth=0.02](-0.9,-2.1)(-0.6,-1.8) 
\psframe[linewidth=0.02](-0.6,-2.1)(-0.3,-1.8)
\psframe[linewidth=0.02](-0.3,-2.1)(0,-1.8) 

\psframe[linewidth=0.02](0,-2.1)(0.3,-1.8)
\psframe[linewidth=0.02](0.3,-2.1)(0.6,-1.8)
\psframe[linewidth=0.02](0.6,-2.1)(0.9,-1.8)
\psframe[linewidth=0.02](0.9,-2.1)(1.2,-1.8)
\psframe[linewidth=0.02](1.2,-2.1)(1.5,-1.8)

\pscircle[fillstyle=solid,fillcolor=black,linewidth=0.0875,linecolor=white](1.75,-1.95){0.14}
\pscircle[fillstyle=solid,fillcolor=black,linewidth=0.0875,linecolor=white](2,-1.95){0.14}
\pscircle[fillstyle=solid,fillcolor=black,linewidth=0.0875,linecolor=white](2.25,-1.95){0.14}

\psframe[linewidth=0.02](2.4,-2.1)(2.7,-1.8)
\psframe[linewidth=0.02](2.7,-2.1)(3,-1.8)
\psframe[linewidth=0.02](3,-2.1)(3.3,-1.8)

\psframe[linewidth=0.02](3.3,-2.1)(3.6,-1.8)
\psframe[linewidth=0.02](3.6,-2.1)(3.9,-1.8)
\psframe[linewidth=0.02](3.9,-2.1)(4.2,-1.8)
\psframe[linewidth=0.02](4.2,-2.1)(4.5,-1.8)
\psframe[linewidth=0.02](4.5,-2.1)(4.8,-1.8)
\psframe[linewidth=0.02](4.8,-2.1)(5.1,-1.8)
\psframe[linewidth=0.02](5.1,-2.1)(5.4,-1.8)
\psframe[linewidth=0.02,fillstyle=solid,fillcolor=darkgray](5.4,-2.1)(5.7,-1.8)
\psframe[linewidth=0.02](5.7,-2.1)(6,-1.8)

\psframe[linewidth=0.02](6,-2.1)(6.3,-1.8)
\psframe[linewidth=0.02](6.3,-2.1)(6.6,-1.8)
\psframe[linewidth=0.02,fillstyle=solid,fillcolor=darkgray](6.6,-2.1)(6.9,-1.8)
\psframe[linewidth=0.02,fillstyle=solid,fillcolor=lightgray](6.9,-2.1)(7.2,-1.8)
\psframe[linewidth=0.02,fillstyle=solid,fillcolor=lightgray](7.2,-2.1)(7.5,-1.8)
\psframe[linewidth=0.02,fillstyle=solid,fillcolor=lightgray](7.5,-2.1)(7.8,-1.8)
\psframe[linewidth=0.02,fillstyle=solid,fillcolor=darkgray](7.8,-2.1)(8.1,-1.8)
\psframe[linewidth=0.02,fillstyle=solid,fillcolor=lightgray](8.1,-2.1)(8.4,-1.8)
\psframe[linewidth=0.02,fillstyle=solid,fillcolor=lightgray](8.4,-2.1)(8.7,-1.8)
\psframe[linewidth=0.02,fillstyle=solid,fillcolor=lightgray](8.7,-2.1)(9,-1.8)

\pscircle[fillstyle=solid,fillcolor=black,linewidth=0.0875,linecolor=white](9.25,-1.95){0.14}
\pscircle[fillstyle=solid,fillcolor=black,linewidth=0.0875,linecolor=white](9.5,-1.95){0.14}
\pscircle[fillstyle=solid,fillcolor=black,linewidth=0.0875,linecolor=white](9.75,-1.95){0.14}

\psframe[linewidth=0.02,fillstyle=solid,fillcolor=darkgray](10,-2.1)(10.3,-1.8)
\psframe[linewidth=0.02](10.3,-2.1)(10.6,-1.8)
\psframe[linewidth=0.02](10.6,-2.1)(10.9,-1.8)
\psframe[linewidth=0.02](10.9,-2.1)(11.2,-1.8)
\psframe[linewidth=0.02,fillstyle=solid,fillcolor=darkgray](11.2,-2.1)(11.5,-1.8)

}

\end{pspicture}
\end{center}
\caption{Structure of TDMA scheme, $\ntone = 2$, $\nttwo = 1$, $L = 4$ and $T=2$.}
\label{fig:pilot_data_illustration_TDMA}
\end{figure*}
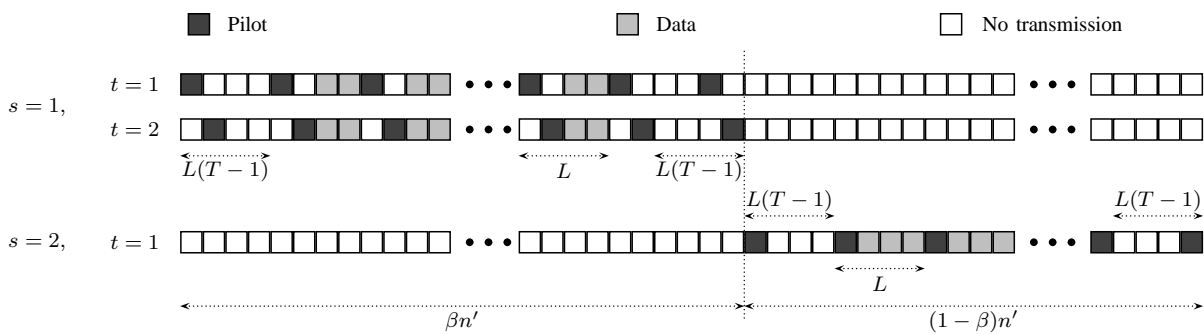

We shall compare the joint-transmission scheme with a time-division multiple-access (TDMA) scheme, where each user transmits its message using the transmission scheme illustrated in Fig.~\ref{fig:pilot_data_illustration_TDMA}. Specifically, during the first $\beta n'$ channel uses (for some $0\leq \beta \leq 1$), User 1 transmits its codeword according to the transmission scheme given in Section \ref{sec:model} (see also Fig.~\ref{fig:pilot_data_illustration_TDMA}), while User 2 is silent. (Here $n'$ is given in \eqref{eq:total_length_mac}.) Then, during the next $(1 - \beta)n'$ channel uses, User 2 transmits its codeword according to the same transmission scheme, while User 1 is silent. In both cases, the receiver guesses the corresponding message $m_s$, $s=1,2$ using a nearest neighbor decoder and pilot-aided channel estimation.

\subsection{The MAC Pre-Log}

Let $R_1^*(\SNR)$, $R_2^*(\SNR)$ and $R_{1 + 2}^* (\SNR)$ be the maximum achievable rate for User 1, the maximum achievable rate for User 2 and the maximum achievable sum-rate, respectively. The achievable-rate region is given by the closure of the convex hull of the set \cite{cover_elements_inf_theory}
\begin{align}
\mathcal{R}  =  \Big \{ R_1(\SNR),R_2 (\SNR) \colon  &R_1(\SNR) < R_1^* (\SNR),\nonumber\\
& R_2 (\SNR) < R_2^* (\SNR),\nonumber\\
& R_1 (\SNR) + R_2 (\SNR) < R_{1+2}^* (\SNR) \Big\}.
\end{align}
We are interested in the pre-logs of $R_1(\SNR)$ and $R_2(\SNR)$, defined as the limiting ratios of $R_1 (\SNR)$ and $R_2 (\SNR)$ to the logarithm of the SNR as the SNR tends to infinity.  Thus, the pre-log region is given by the closure of the convex hull of the set
\begin{align}
\Pi_{\mathcal{R}}  =   \Big\{\Pi_{R_1},\Pi_{R_2}\colon & \Pi_{R_1} < \Pi_{R^*_1},\nonumber\\
& \Pi_{R_2} < \Pi_{R^*_2},  \nonumber \\
& \Pi_{R_1}+\Pi_{R_2}<\Pi_{R^*_{1+2}}\Big\}
\end{align}
where
\begin{align}
 \Pi_{R^*_1}  &\triangleq  \limsup_{\SNR \rightarrow \infty} ~ \frac{R^*_1(\SNR)}{\log \SNR}, \label{def:gmi-pre-log-constraint-R-1} \\ 
 \Pi_{R^*_2}  &\triangleq  \limsup_{\SNR \rightarrow \infty} ~ \frac{R^*_2(\SNR)}{\log \SNR}, \label{def:gmi-pre-log-constraint-R-2} \\ 
 \Pi_{R^*_{1+2}} &\triangleq \limsup_{\SNR \rightarrow \infty} ~ \frac{R^*_{1+2} (\SNR)}{\log \SNR}. \label{def:gmi-pre-log-constraint-R-1-2} \IEEEeqnarraynumspace
\end{align}
The capacity pre-logs $\Pi_{C_1}$, $\Pi_{C_2}$ and $\Pi_{C_{1+2}}$ are defined in the same way but with $R_1^* (\SNR)$, $R_2^* (\SNR)$ and $R_{1+2}^* (\SNR)$ replaced by the respective capacities $C_1 (\SNR)$, $C_2 (\SNR)$ and \linebreak $C_{1+2} (\SNR)$.

We next present our result on the pre-log region of the two-user MIMO fading MAC achievable with the joint-transmission scheme.

\begin{theorem}
\label{th:JTD_pre-log}
 Consider the MIMO fading MAC model \eqref{eq:channel}. Then, the pre-log region achievable with the joint-transmission scheme is the closure of the convex hull of the set
 \begin{align}
 \Bigg  \{    \Pi_{R_1}, \Pi_{R_2}   \colon & \,\, \Pi_{R_1}  <  \min \left( \nr, \ntone \right) \left(1 - \frac{\ntone + \nttwo}{L^*} \right), \nonumber \\
& \,\,\Pi_{R_2}  <  \min \left( \nr, \nttwo \right) \left(1 - \frac{\ntone + \nttwo}{L^*} \right), \nonumber \\
& \,\,\Pi_{R_1} + \Pi_{R_2}   <  \min \left( \nr, \ntone + \nttwo \right) \left(1 - \frac{\ntone + \nttwo}{L^*} \right) \Bigg \} 
\label{eq:thJTD_pre-log}
\end{align}
where $L^* = \left \lfloor \frac{1}{2\lambda_D} \right \rfloor$. 
\end{theorem}
\begin{proof}
 See Section \ref{sec:proof-mac-pre-log}.
\end{proof}

The pre-log region given in Theorem \ref{th:JTD_pre-log} is the largest region achievable with any transmission scheme that uses $(\ntone+\nttwo)/L^*$ of the time for transmitting pilot symbols. Indeed, even if the channel estimator would be able to estimate the fading coefficients perfectly, and even if we could decode the data symbols using a maximum-likelihood decoder, the capacity pre-log region (without pilot transmission) would be given by the closure of the convex hull of the set \cite{Bell_foschini_layered_space-time,telatar_multiantenna_Gaussian,cover_elements_inf_theory}
\begin{align}
 \Big\{(\Pi_{R_1},\Pi_{R_2})\colon &\Pi_{R_1} < \min(\nr, \ntone) \nonumber\\
& \Pi_{R_2} < \min(\nr, \nttwo) \nonumber\\
 & \Pi_{R_1}+\Pi_{R_2} < \min(\nr, \ntone + \nttwo) \Big\}\IEEEeqnarraynumspace
\end{align}
which, after multiplying by $1-(\ntone+\nttwo)/L^*$ in order to account for the pilot symbols, becomes \eqref{eq:thJTD_pre-log}. Thus, in order to improve upon \eqref{eq:thJTD_pre-log}, one would need to design a transmission scheme that employs less than $(\ntone+\nttwo)/L^*$ pilot symbols per channel use.

\begin{remark}[TDMA Pre-Log]
\label{rmrk:TDMA_pre-log}
Consider the MIMO fading MAC model \eqref{eq:channel}. Then, the pre-log region achievable with the TDMA scheme employing nearest neighbor decoding and pilot-aided channel estimation is the closure of the convex hull of the set
\begin{align}
 \Bigg  \{   \Pi_{R_1}, \Pi_{R_2}  \colon & \,\,\Pi_{R_1}  <  \beta \min  \left( \nr, \ntone \right) \left(1 - \frac{\ntone}{L^*} \right), \nonumber \\
 & \,\,\Pi_{R_2}  <  ( 1 - \beta) \min  \left( \nr, \nttwo \right) \left(1 - \frac{\nttwo}{L^*} \right), 0 \leq \beta \leq 1  \Bigg \} \IEEEeqnarraynumspace 
\end{align}
where $L^*= \left \lfloor \frac{1}{2\lambda_D} \right \rfloor$. This follows directly from the pre-log of the point-to-point MIMO fading channel (Theorem \ref{th:pre-log-MIMO}) where the number of transmit antennas from Users 1 and 2 is given by $\ntone$ and $\nttwo$, respectively.
\end{remark}

Note that the sum of the pre-logs $\Pi_{R_1}+\Pi_{R_2}$ is upper-bounded by the capacity pre-log of the point-to-point MIMO fading channel with $(\ntone + \nttwo)$ transmit antennas and $\nr$ receive antennas, since the point-to-point MIMO channel allows for cooperation between the transmitting terminals. While the capacity pre-log of point-to-point MIMO fading channels remains an open problem, the capacity pre-log of point-to-point MISO fading channels is known, cf.~\eqref{eq:MISO-capacity-prelog-BL}. It thus follows from \eqref{eq:MISO-capacity-prelog-BL} that, for $\nr=\ntone=\nttwo=1$, we have
\begin{equation}
\Pi_{R_1}+\Pi_{R_2} \leq \Pi_{C_{1+2}}  = 1-2\lambda_D
\end{equation}
which together with the single-user constraints
\begin{align}
\Pi_{R_1} & \leq  \Pi_{C_1} = 1-2\lambda_D \\
\Pi_{R_2} & \leq  \Pi_{C_2} = 1-2\lambda_D
\end{align}
implies that TDMA achieves the capacity pre-log region of the SISO fading MAC. The next section provides a more detailed comparison between the joint-transmission scheme and TDMA.

\subsection{Joint Transmission versus TDMA}

In this section, we discuss how the joint-transmission scheme performs compared to TDMA. To this end, we compare the sum-rate pre-log $\Pi_{R^*_{1+2}}$ of the joint-transmission scheme (Theorem~\ref{th:JTD_pre-log}) with the sum-rate pre-log of the TDMA scheme employing nearest neighbor decoding and pilot-aided channel estimation (Remark~\ref{rmrk:TDMA_pre-log}) as well as with the sum-rate pre-log of the coherent TDMA scheme, where the receiver has knowledge of the realizations of the fading processes $\{\HRM_{s,k},\,k\in\integ\}$, $s=1,2$. In the latter case, the sum-rate pre-log is given by
\begin{equation}
\label{eq:TDMAcoh}
\Pi_{R^*_{1+2}} = \beta \min(\nr,\ntone) + (1-\beta) \min(\nr,\nttwo).
\end{equation}
The following corollary presents a sufficient condition on $L^*$ under which the sum-rate pre-log of the joint-transmission scheme is strictly larger than that of the coherent TDMA scheme \eqref{eq:TDMAcoh}, as well as a sufficient condition on $L^*$ under which the sum-rate pre-log of the joint-transmission scheme is strictly smaller than the sum-rate pre-log of the TDMA scheme given in Remark~\ref{rmrk:TDMA_pre-log}. Since \eqref{eq:TDMAcoh} is an upper bound on the sum-rate pre-log of any TDMA scheme over the MIMO fading MAC \eqref{eq:channel}, and since the sum-rate pre-log given in Remark~\ref{rmrk:TDMA_pre-log} is a lower bound on the sum-rate pre-log of the best TDMA scheme, it follows that the sufficient conditions presented in Corollary~\ref{cor:JTvsTDMA} hold also for the best TDMA scheme.

\begin{corollary}
\label{cor:JTvsTDMA}
Consider the MIMO fading MAC model \eqref{eq:channel}. The joint-transmission scheme achieves a larger sum-rate pre-log than any TDMA scheme if
\begin{equation}
\label{eq:corJT}
L^* > \frac{\min(\nr,\ntone+\nttwo)(\ntone+\nttwo)}{\min(\nr,\ntone+\nttwo)-\min (\nr,\max(\ntone,\nttwo))}
\end{equation}
where we define $a/0\triangleq \infty$ for every $a>0$. Conversely, the best TDMA scheme achieves a larger sum-rate pre-log than the joint-transmission scheme if
\begin{align}
L^* < \,& \frac{\min(\nr,\ntone+\nttwo)(\ntone+\nttwo)}{\min(\nr,\ntone+\nttwo)-\min(\nr,\ntone,\nttwo)}  \nonumber \\
  &  - \frac{\min(\ntone\nr,\ntone^2,\nttwo\nr,\nttwo^2)}{\min(\nr,\ntone+\nttwo)-\min(\nr,\ntone,\nttwo)}.  \label{eq:corTDMA}
\end{align}
\end{corollary}

Recall that $L^*$ is inversely proportional to the bandwidth of the power spectral density $f_H(\cdot)$, which in turn is inversely proportional to the coherence time of the fading channel. Corollary~\ref{cor:JTvsTDMA} thus demonstrates that the joint-transmission scheme tends to be superior to TDMA when the coherence time of the channel is large. In contrast, TDMA is superior to the joint-transmission scheme when the coherence time of the channel is small. Intuitively, this can be explained by observing that, compared to TDMA, the joint-transmission scheme uses the multiple antennas at the transmitters and at the receiver more efficiently, but requires more pilot symbols to estimate the fading coefficients. Thus, when the coherence time is large, the number of pilot symbols required to estimate the fading is small, so the gain in achievable rate by using the antennas more efficiently dominates the loss incurred by requiring more pilot symbols. On the other hand, when the coherence time is small, the number of pilot symbols required to estimate the fading is large and the loss in achievable rate incurred by requiring more pilot symbols dominates the gain by using the antennas more efficiently.

We next evaluate \eqref{eq:corJT} and \eqref{eq:corTDMA} for some particular values of $\nr$, $\ntone$, and $\nttwo$.

\subsubsection{Receiver employs less antennas than transmitters} 

Suppose that $\nr\leq\min(\ntone,\nttwo)$. Then, the right-hand sides (RHSs) of \eqref{eq:corJT} and \eqref{eq:corTDMA} become $\infty$, so every finite $L^*$ satisfies \eqref{eq:corTDMA}. Thus, if the number of receive antennas is smaller than the number of transmit antennas, then, irrespective of $L^*$, TDMA is superior to the joint-transmission scheme.

\subsubsection{Receiver employs more antennas than transmitters}
Suppose that $\nr\geq\ntone+\nttwo$, and suppose that $\ntone=\nttwo=\nt$. Then, \eqref{eq:corJT} and \eqref{eq:corTDMA} become
\begin{equation}
\label{eq:JT_SIMO}
L^* > 4\nt
\end{equation}
and
\begin{equation}
\label{eq:TDMA_SIMO}
L^* < 3\nt.
\end{equation}
Thus, if $L^*$ is greater than $4\nt$, then the joint-transmission scheme is superior to TDMA. In contrast, if $L^*$ is smaller than $3\nt$, then TDMA is superior. This is illustrated in Fig.~\ref{fig:simo-mac-plot} for the case where $\nr=2$ and $\ntone=\nttwo=1$. Note that if $L^*$ is between $3\nt$ and $4\nt$, then the joint-transmission scheme is superior to the TDMA scheme presented in Remark \ref{rmrk:TDMA_pre-log}, but it may be inferior to the best TDMA scheme.

\begin{figure*}[t]
\mbox{
\begin{psfrags}
 \psfrag{a}{\tiny $1- \frac{1}{L^*}$}
\psfrag{b}{\tiny $1 - \frac{2}{L^*}$}
\psfrag{c}{\tiny $\Pi_{R_1}$}
\psfrag{d}{\tiny $\Pi_{R_2}$}
\psfrag{e}{\tiny $1$}
\subfloat[$ L^* < 3 $] {\includegraphics[width=0.35\textwidth]{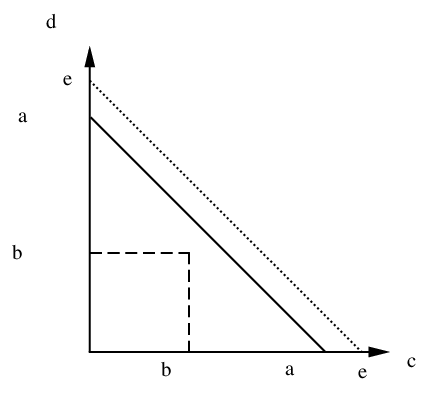}
\label{subfig:simo-mac-alpha-p3-to-p2} }
\end{psfrags}
\begin{psfrags}
\psfrag{a}{\tiny $1-\frac{1}{L^*}$}
\psfrag{b}{\tiny $1 - \frac{2}{L^*}$}
\psfrag{c}{\tiny $\Pi_{R_1}$}
\psfrag{d}{\tiny $\Pi_{R_2}$}
\psfrag{e}{\tiny $1$}
 \subfloat[$ L^* > 4$]{ \includegraphics[width=0.35\textwidth]{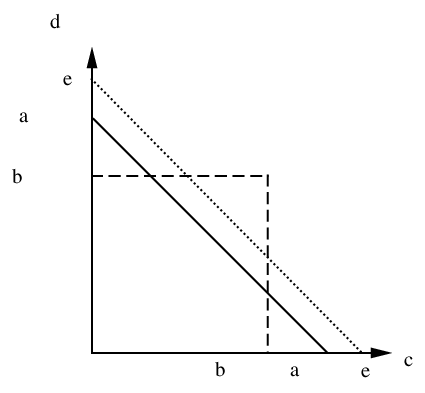}
\label{subfig:simo-mac-alpha-p4}
} 
\end{psfrags}\hspace{-25mm}
\subfloat{\includegraphics[width=0.4\textwidth]{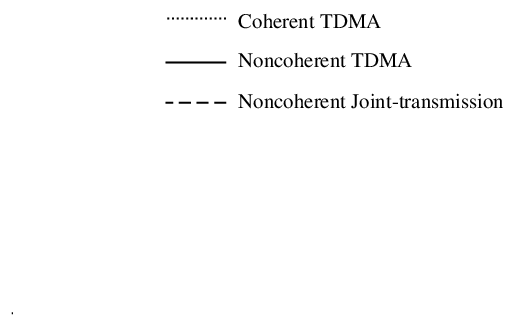}}}
\caption{Pre-log regions for a fading MAC with $\nr=2$ and $\ntone=\nttwo=1$ for different values of $L^*$. Depicted are the pre-log region for the joint-transmission scheme as given in Theorem~\ref{th:JTD_pre-log} (dashed line), the pre-log region of the TDMA scheme as given in Remark~\ref{rmrk:TDMA_pre-log} (solid line), and the pre-log region of the coherent TDMA scheme \eqref{eq:TDMAcoh} (dotted line).}
\label{fig:simo-mac-plot}
\end{figure*}

\subsubsection{A case in between}
Suppose that $\nr\leq\ntone+\nttwo$ and $\nttwo<\nr\leq\ntone$. Then, \eqref{eq:corJT} becomes
\begin{equation}
L^* > \infty
\end{equation}
and \eqref{eq:corTDMA} becomes
\begin{equation}
L^* < \nttwo+\frac{\nr\ntone}{\nr-\nttwo}.
\end{equation}
Thus, in this case the joint-transmission scheme is always inferior to the coherent TDMA scheme \eqref{eq:TDMAcoh}, but it can be superior to the TDMA scheme in Remark \ref{rmrk:TDMA_pre-log}.

\subsection{Typical Values of $L^*$}
We briefly discuss the range of values of $L^*$ that may occur in practical scenarios. To this end, we first recall that $L^*\leq \lfloor 1/(2\lambda_D) \rfloor$, and that $\lambda_D$ is the bandwidth of the fading power spectral density $f_H(\cdot)$, which can be associated with the Doppler spread of the channel as \cite{IEEE:etkin:degreeofffreedomMIMO}
\begin{equation}
\lambda_D = \frac{f_m}{W_c}.
\end{equation}
Here $f_m$ is the maximum Doppler shift given by
\begin{equation}
f_m = \frac{v}{c}f_c
\end{equation}
where $v$ is the speed of the mobile device, $c = 3\cdot10^8$ m/s is the speed of light and $f_c$ is the carrier frequency; and $W_c$ is the coherence bandwidth of the channel approximated as \cite{rappaport_wireless-comm,IEEE:etkin:degreeofffreedomMIMO}
\begin{equation}
W_c \approx \frac{1}{5\sigma_\tau}
\end{equation}
where $\sigma_\tau$ is the delay spread. Following the order of magnitude computations of Etkin and Tse \cite{IEEE:etkin:degreeofffreedomMIMO}, we determine typical values of $\lambda_D$ for indoor, urban, and hilly area environments and for carrier frequencies ranging from 800 MHz to 5 GHz and tabulate the results in Table \ref{table:typical_values_L}. 

\begin{table}[t]

\begin{center}
{\small 
\begin{tabular}{lcccccc}

\toprule
Environment 	& Delay spread $\sigma_\tau$ & Mobile speed $v$ 	& $\lambda_D \approx 5 \sigma_\tau \frac{v}{c} f_c$					&$L^*$\\
\midrule
Indoor 		& 10 -- 100 ns	& 5 km/h	&  $2\cdot 10^{-7}$ -- $10^{-5}$						&$5\cdot10^4$ -- $2.5\cdot10^{6}$	\\
Urban		& 1 -- 2 $\mu$s	& 5 km/h 	& $2\cdot 10^{-5}$ -- $2\cdot 10^{-4}$						& $2.5 \cdot 10^3$ -- $2.5 \cdot 10^4$	\\
Urban		& 1 -- 2 $\mu$s	& 75 km/h	& $2\cdot 10^{-4}$ -- $0.004$						& 125 -- $2.5 \cdot 10^3$	\\
Hilly area		& 3 -- 10 $\mu$s	& 200 km/h 	& $0.002$ -- $0.05$ & 10 -- 250 \\
\bottomrule
\end{tabular}
}
\end{center}
\caption{Typical values of $L^*$ for various environments with $f_c$ ranging from 800 MHz to 5 GHz. The values of the delay spread are taken from \cite{rappaport_wireless-comm,IEEE:etkin:degreeofffreedomMIMO} for indoor and urban environments and from \cite{saunders_propagation_wireless_comm} for hilly area environments.}
\label{table:typical_values_L}
\end{table}

For indoor environments and mobile speeds of 5 km/h, we have that $L^*$ is typically larger than $5 \cdot 10^4$. For urban environments, $L^*$ is typically larger than $2.5\cdot 10^3$ for mobile speeds of 5 km/h and larger than $125$ for mobile speeds of 75 km/h. For hilly area environments and mobile speeds of 200 km/h, $L^*$ ranges typically from $10$ to $250$. Thus, for most practical scenarios, $L^*$ is typically large. It therefore follows that, if $\nr\geq\ntone+\nttwo$, the condition \eqref{eq:corJT} is satisfied unless $\ntone+\nttwo$ is very large. For example, if the receiver employs more antennas than the transmitters, and if $\ntone=\nttwo=\nt$, then $L^* > 4\nt$ is satisfied even for urban environments and mobile speeds of 75 km/h, as long as $\nt<30$. Only for hilly area environments and mobile speeds of 200 km/h, this condition may not be satisfied for a practical number of transmit antennas. Thus, if the number of antennas at the receiver is sufficiently large, then the joint-transmission scheme is superior to TDMA in most practical scenarios. On the other hand, if $\nr\leq\min(\ntone,\nttwo)$, then TDMA is always superior to the joint-transmission scheme, irrespective of how large $L^*$ is. This suggests that one should use more antennas at the receiver than at the transmitters.

\section{Proof of Theorem \ref{th:pre-log-MIMO}}
\label{sec:proof-theorem-1}

Theorem \ref{th:pre-log-MIMO} is proven as follows. We first characterize the estimation error from the linear interpolator \eqref{eq:LMMSEestimation}. We then compute the rates achievable with the communication scheme described in Section \ref{sec:model}. Finally, we analyze the pre-log corresponding to these rates.

\subsection{Linear Interpolator}
\label{subsec:channelestimator}

We first note that the estimate of $H_k (r,t)$ is given by \eqref{eq:LMMSEestimation}, namely,
\begin{equation}
\hat H_k^{(T)} (r,t) = \sum^{k + T L}_{\substack{ k' = k - TL:\\ k' \in \mathcal{P}}}  a_{k'} (r,t) Y_{k'} (r), \qquad k \in \mathcal{D}. \label{eq:fading-estimate-pilot-only-observation}
\end{equation}
We denote the interpolation error by $E^{(T)}_k (r,t) = H_k (r,t) - \hat H_k^{(T)} (r,t)$. 

For future reference, and for any $k \in \integ$, we express $k = jL + \ell$, so $\ell=k \: {\rm mod} \: L$. Assuming that the first pilot symbol is transmitted at $k=0$, it follows that $\ell=0,\dotsc,\nt-1$ for $k \in \mathcal{P}$ and $\ell=\nt,\dotsc,L-1$ for $k \in \mathcal{D}$. The statistical properties of the channel estimator for a given window size $T$ are summarized in the following lemma. 

\begin{lemma}
\label{lemma:interpolator-property-fix-T}
For a given $T$, the linear interpolator \eqref{eq:fading-estimate-pilot-only-observation} has the following properties.
\begin{enumerate}

\item  For each $t=1,\dotsc,\nt$, $r=1,\dotsc,\nr$ and $\ell = \nt,\dotsc,L-1$, the estimate $\hat H^{(T)}_{jL + \ell} (r,t)$ and the corresponding estimation error $E^{(T)}_{jL + \ell} (r,t)$ are independent zero-mean complex-Gaussian random variables.

\item  

\begin{enumerate}
\item For a given transmit antenna $t$ and $\ell \in \{\nt,\dotsc,L-1 \}$, the $\nr$ processes 
\begin{equation*}
\{ (\hat H_{jL+\ell}^{(T)} (1,t), E_{jL+\ell}^{(T)} (1,t)), \: j \in \integ \},\dotsc,\{ (\hat H_{jL+\ell}^{(T)} (\nr,t), E_{jL+\ell}^{(T)} (\nr,t)), \: j \in \integ \}
\end{equation*}
are independent and have the same law. 

\item For a given receive antenna $r$ and  $\ell \in \{\nt,\dotsc,L-1 \}$, the $\nt$ processes 
\begin{equation*}
\{ (\hat H_{jL+\ell}^{(T)} (r,1), E_{jL+\ell}^{(T)} (r,1)), \: j \in \integ \}, \dotsc,\{ (\hat H_{jL+\ell}^{(T)} (r,\nt), E_{jL+\ell}^{(T)} (r,\nt)), \: j \in \integ \}
\end{equation*}
are independent but have different laws. 
\end{enumerate}

\item For each $\ell = \nt,\dotsc,L-1$, the process $ \{ (\hat \HRM^{(T)}_{jL+\ell}, \: \HRM_{jL + \ell}, \: \zrv_{jL + \ell},  \:  \xrv_{jL + \ell}), \: j \in \integ  \}$  is jointly stationary and ergodic.

\item For $\ell = \nt,\dotsc,L-1$, it holds that
\begin{equation}
\Expect \left[  \zrv^\dagger_\ell \hat \HRM^{(T)}_\ell  \xrv_\ell  \right] = 0 \label{eq:uncorrelated-noise-Expec-D-ell-1}
\end{equation}
where $\hermi{(\cdot)}$ denotes the conjugate transpose.

\end{enumerate}

\begin{proof}
See Appendix \ref{sec:proof-lemma:interpolator-property-fix-T}.
\end{proof}

\end{lemma}

\subsection{Achievable Rates and Pre-Logs}
\label{subsec:achievable-rates}

In the following proof, we only consider the case where $\nt = \nr$. The more general case of $\nt \neq \nr$ follows then by employing only $\nr$ transmit antennas or by ignoring $\nr-\nt$ antennas at the receiver. This yields a lower bound on the maximum achievable rate and does not incur a loss with respect to the pre-log. Indeed, it can be shown that the nearest neighbor decoder described in Section \ref{sec:model} achieves the pre-log $\min(\nr,\nt)$. Thus, increasing $\nt$ beyond $\nr$ or $\nr$ beyond $\nt$ does not improve the pre-log achievable by such a decoder. In fact, increasing $\nt$ beyond $\nr$ requires the transmission of more pilot symbols and does therefore even reduce the pre-log achievable with the communication system described in Section \ref{sec:model}. 

To prove Theorem \ref{th:pre-log-MIMO}, we analyze the generalized mutual information (GMI) \cite{IEEE:mehrav:oninformationrates_mismatched} for the channel and communication scheme in Section \ref{sec:model}. The GMI, denoted by $\gmi_T (\SNR)$, specifies the highest information rate for which the average probability of error, averaged over the ensemble of i.i.d. Gaussian codebooks, tends to zero as the codeword length $n$ tends to infinity (see \cite{IEEE:lapidoth:nearestneighbournongaussian,IEEE:lapidoth:fadingchannels_howperfect,IEEE:weingarten:gaussiancodes}  and references therein). The GMI for stationary Gaussian fading channels employing nearest neighbor decoding has been evaluated in \cite{IEEE:lapidoth:fadingchannels_howperfect,IEEE:weingarten:gaussiancodes} for the case where a genie provides the receiver with an estimate of the fading process. However, the estimate considered in \cite{IEEE:lapidoth:fadingchannels_howperfect,IEEE:weingarten:gaussiancodes} is assumed to be jointly stationary ergodic with $\{ (\HRM_k, \xrv_k, \zrv_k), \: k \in \integ \}$, which is not satisfied by $\{ \hat \HRM^{(T)}_k, \: k \in \mathcal{D} \}$. We therefore need to adapt the work in  \cite{IEEE:lapidoth:fadingchannels_howperfect,IEEE:weingarten:gaussiancodes} to our channel model. For completeness, we present all the main steps here, even though they are very similar to the ones in\cite{IEEE:lapidoth:fadingchannels_howperfect,IEEE:weingarten:gaussiancodes}. 

We prove Theorem \ref{th:pre-log-MIMO} as follows:
\begin{enumerate}
\item We compute a lower bound on $\gmi_T (\SNR)$ for a fixed window size $T$. 
\item We analyze the behavior of this lower bound as $T$ tends to infinity. 
\item We evaluate the limiting ratio of this lower bound to $\log \SNR$ as $\SNR$ tends to infinity. 
\end{enumerate}

\vspace{5mm}

\subsubsection{\underline{$\gmi_T (\SNR)$ for a fixed $T$} } 

We analyze the GMI for a fixed $T$ using a random coding upper bound on the average error probability. Note that due to the symmetry of the codebook construction, it suffices to consider the error behavior, conditioned on the event that message $1$ was transmitted. Let $\mathcal{E} (m')$ denote the event that $D(m') \leq D(1)$. The ensemble-average error probability, where the average is over the ensemble of i.i.d. Gaussian codes, corresponding to message $m=1$ is thus given by 
\begin{equation}
\bar P_e (1) = \Pr \left \{ \bigcup_{m' \neq 1} \mathcal{E} (m') \right\}. \label{eq:rc-err-m1}
\end{equation}

To evaluate the GMI from the RHS of \eqref{eq:rc-err-m1}, we define some useful quantities in the following. Recall the channel and the transmission model in Section \ref{sec:model}. Without loss of generality, assume that the first pilot vector is transmitted at time $k=0$. Define $F(\SNR)$ as
\begin{align}
F(\SNR) &\triangleq \nr  + \frac{\SNR}{(L - \nt) \nt} \sum^{(L-1)}_{\ell = \nt}    \Expect \left[ \left \| \ERM^{(T)}_\ell \right \|^2_F  \right]  \label{eq:F-SNR-defined}
\end{align}
where $\ERM^{(T)}_\ell$ is a random matrix with element at row $r$ and column $t$ given by $E^{(T)}_{\ell}(r,t)$, and where $\| \cdot \|_F$ denotes the Frobenius norm. Further define a typical set 
\begin{align}
 \mathcal{T}_{\delta}   \triangleq  \Bigg \{ &  \left(\xv_k, \yv_k, \hat \HM^{(T)}_k \right), k = 0,\dotsc, n'-1: \nonumber \\
 &  \qquad \left | \frac{1}{n} \sum_{k \in \mathcal{D}^{(n')} } \left
 \|\yv_k - \sqrt{\frac{\SNR}{\nt}} ~ \hat \HM^{(T)}_k \xv_k \right\|^2  - F(\SNR) \right | < \delta  \Bigg  \}  \label{eq:typical_set_delta}
\end{align}
with $\mathcal{D}^{(n')} = \{0,\dotsc,n' -1 \} \cap \mathcal{D}$ as provided in \eqref{eq:Dnprime-defined}, and some $\delta > 0$, where we have recalled $n'$ in \eqref{eq:total_length}, namely 
\begin{equation}
n' = n_{\rm p} + n + n_{\rm g}.
\end{equation}
Then, we have the following convergence as $n$ tends to infinity. 

\begin{lemma}
\label{lemma:F-SNR-typical-set}
 For the communication scheme described in Section \ref{sec:model}, we have that
\begin{equation}
\lim_{n \to \infty} \Pr \Big \{   \left( \xrv^{n'}, \yrv^{n'}, \hat \HRM^{ (T),n'}  \right) \in \mathcal{T}_{\delta}  \Big \}  = 1, \qquad \forall \delta > 0 \label{eq:lim-Pr-Typ-n-infinity}
\end{equation}
where we have used the notation $U^{n'}$ to denote  the sequence $U_0,\dotsc, U_{n' - 1}$.
\end{lemma}
\begin{proof}
We have 
\begin{align}
& \lim_{n \to \infty}  \frac{1}{n}  \sum_{k \in \mathcal{D}^{(n')}} \left \| \yv_k - \sqrt{\frac{\SNR}{\nt}} ~\hat \HM^{(T)}_k  \xv_k  \right\|^2  \nonumber \\
&  = \lim_{n \to \infty}  \frac{1}{n}  \sum_{k \in \mathcal{D}^{(n')}} \left \| \sqrt{\frac{\SNR}{\nt}} ~  \left(\HM_k  - \hat \HM^{(T)}_k  \right) \xv_k  + \zv_k \right\|^2   \\
&   =  \frac{1}{ L - \nt }  \sum_{\ell = \nt}^{L-1}  \lim_{n \to \infty}  \frac{L - \nt}{n} \sum_{j = 0}^{ \frac{n}{L-\nt} - 1} \left \| \sqrt{\frac{\SNR}{\nt}}    \left(\HM_{jL + \ell}  - \hat \HM^{(T)}_{jL + \ell}  \right) \xv_{jL + \ell}  + \zv_{jL + \ell} \right\|^2  \\
&   =  \frac{1}{L-\nt} \sum_{\ell = \nt}^{L-1} \Expect \left[ \left \| \sqrt{\frac{\SNR}{\nt}} \left( \HRM_\ell - \hat \HRM^{(T)}_\ell \right) \bar \xrv_\ell + \zrv_\ell \right \|^2  \right ], \qquad \mbox{almost surely} \label{eq:ergodicity-F} \\
&   = \frac{1}{L - \nt }\sum^{L-1}_{\ell = \nt} \left( \nr + \frac{\SNR}{\nt} \Expect \left[ \left \| \ERM^{(T)}_\ell \bar \xrv_\ell  \right \|^2  \right]\right) \label{eq:trace-hat-h-x-z-zero-cond} \\
&   = F(\SNR). \label{eq:xrv-ell-gauss-rv-zero-mean}
\end{align}
Herein \eqref{eq:ergodicity-F} follows from Part 3) of Lemma \ref{lemma:interpolator-property-fix-T} and the ergodic theorem \cite[Chap. 7]{durrett:probability}; \eqref{eq:trace-hat-h-x-z-zero-cond}  follows from Part  4) of Lemma \ref{lemma:interpolator-property-fix-T}; and \eqref{eq:xrv-ell-gauss-rv-zero-mean} follows since $\bar \xrv_\ell$ has zero mean and covariance matrix $\IM_\nt$, and is independent from $\ERM_\ell^{(T)}$ (since $\{ \ERM^{(T)}_k,\: k \in \mathcal{D}\}$ is a function of $\{ (\HRM_k, \zrv_k), k \in \integ \}$). It thus follows that 
\begin{equation}
\lim_{n \to \infty}  \frac{1}{n}  \sum_{k \in \mathcal{D}^{(n')}} \left \| \yv_k - \sqrt{\frac{\SNR}{\nt}} ~\hat \HM^{(T)}_k  \xv_k   \right\|^2
\end{equation}
converges to $F(\SNR)$ almost surely, which in turn implies that it also converges in probability, thus proving \eqref{eq:lim-Pr-Typ-n-infinity}. 
\end{proof}

Considering the typical set  \eqref{eq:typical_set_delta} and following the derivation in  \cite{IEEE:lapidoth:fadingchannels_howperfect,IEEE:weingarten:gaussiancodes}, $\bar P_e (1)$ in \eqref{eq:rc-err-m1} can be upper-bounded as 
\begin{align}
\bar P_e(1)
\leq & e^{nR} \cdot \Pr \left \{ \left. \frac{1}{ n } \cdot D(m') < F(\SNR) + \delta  \right |  \left( \xrv^{n'} (1), \yrv^{n'}, \hat \HRM^{(T), n'} \right) \in \mathcal{T}_{\delta}   \right \} \nonumber \\
&   + \Pr \left \{ \left( \xrv^{n'} (1), \yrv^{n'},  \hat \HRM^{(T), n'} \right)  \in \mathcal{T}^c_{\delta}  \right \} , \quad m' \neq 1  \label{eq:bounding-typical-atypical}
\end{align}
where $\mathcal{T}^c_\delta$ denotes the complement of $\mathcal{T}_\delta$. It follows from Lemma \ref{lemma:F-SNR-typical-set}  that the second term on the RHS of \eqref{eq:bounding-typical-atypical} can be made arbitrarily small by letting $n$ tend to infinity. 

The GMI characterizes the rate of exponential decay of the expression
\begin{equation}
\Pr \left \{ \left. \frac{1}{ n } \cdot D(m') < F(\SNR) + \delta  \right | \left( \xrv^{n'} (1), \yrv^{n'}, \hat \HRM^{(T), n'} \right) \in \mathcal{T}_{\delta}   \right \}, \quad m' \neq 1
\end{equation}
as $n \to  \infty$ \cite{IEEE:lapidoth:fadingchannels_howperfect,IEEE:weingarten:gaussiancodes}. The computation of the GMI requires the conditional log moment-generating function of the metric $D(m')$ associated with the wrong message output $m'\neq 1$---conditioned on the channel outputs and on the fading estimates---which we shall denote by $\kappa_{n} (\theta,\yv^{n'}, \hat \HM^{(T),n'} )$, i.e.,
\begin{align}
\kappa_{n} \left (\theta,\yv^{n'}, \hat \HM^{(T),n'} \right) & =  \log \Expect \left[ \left.  {\rm exp} \left(\frac{\theta}{n} \sum_{k \in \mathcal{D}^{(n')}}  D_k (m')   \right) \right| \left\{(\yv_k, \hat \HM^{(T)}_{k}), \: k \in {\mathcal{D}^{(n')}}   \right\} \right]  \label{eq:conditional-log-mgf-for-n}.
\end{align}
Here we define
\begin{equation}
D_k (m') \triangleq \left \| \yv_k  -  \sqrt{ \frac{\SNR}{ \nt}} \: \hat \HM^{(T)}_k \xv_k {(m')} \right \|^2.
\end{equation}
Proceeding along the lines of  \cite{IEEE:lapidoth:fadingchannels_howperfect,IEEE:weingarten:gaussiancodes}, we can express the conditional log moment-generating function in \eqref{eq:conditional-log-mgf-for-n} as the sum of conditional log moment-generating functions for the individual vector metrics $D_k (m')$, $k \in  \mathcal{D}^{(n')}$, i.e.,  
\begin{align}
& \kappa_{n} \left (\theta,\yv^{n'}, \hat \HM^{(T),n'} \right) \nonumber \\
& \quad = \sum_{k \in  \mathcal{D}^{(n')}} \log \mathsf{E} \left[ \left.  {\rm exp} \left(\frac{\theta}{n}   D_k (m')   \right) \right| \yv_k, \hat \HM^{(T)}_{k}\right] \label{eq:conditional-log-mgf-n-to-sum-mgf} \\
& \quad = \sum_{k \in  \mathcal{D}^{(n')}} \left( \frac{\theta}{n} \yv_{k}^\dagger \left( \mathsf{I}_\nr - \frac{\theta}{n} \frac{\SNR}{\nt} \hat \HM_{k}^{(T)} \hat \HM_{k}^{\dagger (T)}  \right)^{-1} \yv_{k} - \log {\rm det} \left( \mathsf{I}_\nr - \frac{\theta}{n} \frac{\SNR}{\nt} \hat \HM_{k}^{(T)} \hat \HM_{k}^{\dagger (T)}  \right) \right). \label{eq:sum-conditional-log-mgf}
\end{align}
We then have that for all $\theta < 0$
\begin{align}
&   \lim_{n \rightarrow \infty} ~\frac{ 1}{n} \cdot \kappa_{n} \left (n \theta,\yv^{n'}, \hat \HM^{(T),n'} \right)   \nonumber \\
& \quad = \lim_{n \to \infty}  \frac{1}{n} \sum_{k \in \mathcal{D}^{(n')}} \theta \yv_{k}^\dagger \left(  \mathsf{I}_\nr - \theta \frac{\SNR}{\nt} \hat \HM_{k}^{(T)} \hat \HM_{k}^{\dagger (T)}  \right)^{-1} \yv_{k}   \nonumber \\
& \qquad \qquad \qquad - \lim_{n \to \infty}  \frac{1}{n} \sum_{k \in  \mathcal{D}^{(n')}} \log {\rm det} \left( \mathsf{I}_\nr - \theta \frac{\SNR}{\nt} \hat \HM_{k}^{(T)} \hat \HM_{k}^{\dagger (T)}  \right)  \label{eq:log-mgf-n-to-infinity} \\
& \quad =  \frac{1}{L - \nt}  \sum_{\ell = \nt}^{L-1}  \lim_{n \to \infty}   \frac{L - \nt}{n} \sum_{j = 0}^{ \frac{n}{L-\nt} - 1} \theta \yv_{jL + \ell}^\dagger \left(  \mathsf{I}_\nr - \theta \frac{\SNR}{\nt} \hat \HM_{jL + \ell}^{(T)} \hat \HM_{jL + \ell}^{\dagger (T)}  \right)^{-1} \yv_{jL + \ell}   \nonumber \\
&  \qquad \qquad  - \frac{1}{L - \nt}  \sum_{\ell = \nt}^{L-1}  \lim_{n \to \infty} \: \frac{L - \nt}{n} \sum_{j = 0}^{ \frac{n}{L-\nt} - 1} \log {\rm det} \left( \mathsf{I}_\nr - \theta \frac{\SNR}{\nt} \hat \HM_{jL + \ell}^{(T)} \hat \HM_{jL + \ell}^{\dagger (T)}  \right)  \\
& \quad = \frac{1}{L - \nt} \sum^{L- 1}_{\ell= \nt} \Expect \left[ \theta \yrv_\ell^\dagger \cdot \left( \mathsf{I}_\nr - \theta \frac{\SNR}{ \nt} \hat \HRM_{\ell}^{(T)} \hat \HRM_{\ell}^{\dagger(T)}  \right)^{-1} \cdot \yrv_\ell \right]   \nonumber \\
&    \qquad \qquad \quad  - \: \frac{1}{L - \nt} \sum^{L- 1}_{\ell=\nt} \Expect \left[ \log {\rm det} \left(\mathsf{I}_\nr - \theta \frac{\SNR}{\nt} \hat \HRM_{\ell}^{(T)} \hat \HRM_{\ell}^{\dagger(T)} \right)   \right], \: \: \mbox{almost surely} \label{eq:converge-log-mgf}  \\
& \quad \triangleq \kappa (\theta, \SNR) \nonumber
\end{align}
where the last step should be regarded as the definition of $\kappa (\theta, \SNR)$. The convergence in \eqref{eq:converge-log-mgf}  is due to the ergodicity of $\{(\yrv_{jL + \ell},\: \hat \HRM_{jL + \ell}^{(T)}), \: j \in \integ\}$, $\ell=\nt,\dotsc,L-1$ (see Part 3) of Lemma \ref{lemma:interpolator-property-fix-T}) and the ergodic theorem. 

Following the same steps as  in \cite{IEEE:lapidoth:fadingchannels_howperfect, IEEE:weingarten:gaussiancodes}, we can then show that for all $\delta' > 0$, the ensemble-average error probability can be bounded as 
\begin{equation}
 \bar P_e (1) \leq {\rm exp} (nR) {\rm exp} \left( - n\left( I^{\rm gmi}_T (\SNR) - \delta' \right) \right) + \varepsilon (\delta',n) \label{eq:ens-error-R-gmi}
\end{equation}
for some $\varepsilon (\delta',n)$ satisfying
\begin{equation}
\lim_{n \to \infty} \varepsilon (\delta',n) = 0, \quad \delta' > 0. 
\end{equation}
On the RHS of \eqref{eq:ens-error-R-gmi}, $I^{\rm gmi}_T (\SNR)$ denotes the GMI as a function of $\SNR$ for a fixed $T$, which is given by 
\begin{equation}
I^{\rm gmi}_T ( \SNR) =  \frac{L - \nt}{L} \left( \sup_{\theta < 0} ~ \left(\theta F(\SNR) - \kappa (\theta, \SNR)  \right) \right). \label{eq:GMI-gartner-ellis}
\end{equation}
Herein the pre-factor $(L-\nt)/L$ equals the fraction of time used for data transmission. The bound \eqref{eq:ens-error-R-gmi} implies that for rates below $\gmi_T (\SNR)$, the communication scheme described in Section \ref{sec:model} has vanishing error probability as $n$ tends to infinity. Combining \eqref{eq:F-SNR-defined} and \eqref{eq:converge-log-mgf} with \eqref{eq:GMI-gartner-ellis} yields 
\begin{align}
&   I^{\rm gmi}_T  (\SNR) \nonumber \\
&  =  \sup_{\theta < 0}  \frac{1}{L}  \sum^{L-1}_{\ell = \nt} \bigg \{ \theta \left(  \nr + \frac{\SNR}{\nt} \Expect \left[\left \| \mathbb E_{\ell}^{(T)} \right \|^2_F \right] \right)  + \Expect \left[ \log {\rm det} \left(\IM_{\nr} - \theta \frac{\SNR}{\nt} \hat \HRM_{\ell}^{(T)} \hat \HRM_{\ell}^{\dagger(T)} \right) \right]   \nonumber \\
&  \qquad \qquad \qquad \qquad    -  \Expect \left[ \theta \hermi{\yrv}_{\ell} \left( \IM_{\nr} - \theta \frac{\SNR}{\nt} \hat \HRM_{\ell}^{(T)} \hat \HRM_{\ell}^{\dagger (T)} \right)^{-1} \yrv_{\ell} \right] \bigg \}. \label{eq:gmi-closed-form-with-sup-over-theta}
\end{align}

Following the steps used in \cite[App. D]{IEEE_IT_Asyhari_Nearest_neighbour}, it can be shown that for $\theta < 0$
\begin{equation}
 - \Expect \left[ \theta \hermi{\yrv}_{ \ell} \left( \IM_{\nr} - \theta \frac{\SNR}{\nt} \hat \HRM_{\ell}^{(T)} \hat \HRM_{\ell}^{\dagger (T)} \right)^{-1} \yrv_{\ell} \right] \geq 0. \label{eq:positive-term-ignored-gmi}
\end{equation}
As observed in \cite[App. D]{IEEE_IT_Asyhari_Nearest_neighbour}, a good lower bound on $\gmi_T (\SNR)$ for high SNR follows by choosing 
\begin{equation}
 \theta = \frac{-1}{\nr + \SNR \, \nr \epsilon^2_{*,T}} \label{eq:suboptimal-theta}
\end{equation}
where 
\begin{equation}
 \epsilon^2_{*,T} = \max_{ \substack{r = 1,\dotsc,\nr,  \\ t =  1,\dotsc,\nt, \\ \ell = \nt,\dotsc,L-1}} \:  \epsilon^2_{\ell,T} (r,t). \label{eq:epsilon-sq-max-T}
\end{equation}
Hence, substituting the choice of $\theta$ in \eqref{eq:suboptimal-theta} and applying \eqref{eq:positive-term-ignored-gmi} to the RHS of \eqref{eq:gmi-closed-form-with-sup-over-theta} yields
\begin{equation}
  I^{\rm gmi}_T  (\SNR) \geq \frac{1}{L}  \sum^{L-1}_{\ell = \nt} \left \{ \Expect \left[ \log {\rm det} \left(\IM_{\nr} + \frac{\SNR}{\nt\nr +  \nt \nr \SNR \epsilon^2_{*,T}} \hat \HRM_\ell^{(T)} \hat \HRM_\ell^{\dagger (T)} \right) \right] - 1  \right \}. \label{eq:GMILB_Tfixed}
\end{equation}

\subsubsection{\underline{Achievable Rates as $T \to \infty$}}

We next analyze the RHS of \eqref{eq:GMILB_Tfixed} in the limit as  $T$ tends to infinity. To this end, we note that, for $L \leq \frac{1}{2\lambda_D}$, the variance of the interpolation error tends to \eqref{eq:fadingestimate_error}, namely
\begin{equation}
\epsilon_\ell^2 (t) =  1 - \int^{1/2}_{-1/2} \frac{\SNR \left[ f_H (\lambda) \right]^2}{\SNR f_H (\lambda) + L \nt} d \lambda  \label{eq:epsilon-sq-irrespective-ell-t}
\end{equation}
irrespective of $\ell$ and $t$. We shall therefore denote the variance of the interpolation error $\epsilon^2_\ell (t)$ by $\epsilon^2$. Note that for a fixed $T$, the entries of 
\begin{equation}
\frac{1}{\sqrt{\nt \nr + \nt \nr \SNR \epsilon^2_{*,T}}} \hat \HRM_\ell^{(T)} 
\end{equation}
are independent but not i.i.d., which follows from Part 2) of Lemma \ref{lemma:interpolator-property-fix-T}. However, as $T$ tends to infinity, their distribution becomes identical  due to \eqref{eq:epsilon-sq-irrespective-ell-t} and hence they converge in distribution to 
\begin{equation}
\frac{1}{\sqrt{\nt \nr + \nt \nr \SNR \epsilon^2_{*,T}}} \hat \HRM_{\ell}^{(T)} \quad \stackrel{d}{\longrightarrow} \quad \frac{1}{\sqrt{\nt \nr + \nt \nr \SNR \epsilon^2 }} \bar \HRM 
\end{equation}
where the entries of $\bar \HRM$ are i.i.d. complex-Gaussian random variables with zero mean and variance $(1 - \epsilon^2)$. 

Note that 
\begin{equation}
\log {\rm det} \left(\IM_{\nr} + \frac{\SNR}{\nt \nr + \nt \nr \SNR \epsilon^2_{*,T}}   \hat \HRM_{\ell}^{(T)}  \hat \HRM_{\ell}^{\dagger (T)} \right) \geq 0
\end{equation}
is a continuous function with respect to the entries of the matrix 
\begin{equation}
\frac{1}{\nt \nr + \nt \nr\SNR \epsilon^2_{*,T}} \hat \HRM_{\ell}^{(T)}  \hat \HRM_{\ell}^{\dagger (T)}.
\end{equation}
It therefore follows from Portmanteau's lemma \cite{vdvaart_weakconvergence} that, as $T \rightarrow \infty$, the RHS of  \eqref{eq:GMILB_Tfixed} can be lower-bounded by
\begin{align}
&  \lim_{T \rightarrow \infty} ~ \frac{1}{L}  \sum^{L-1}_{\ell = \nt} \left \{ \Expect \left[ \log {\rm det} \left(\IM_{\nr} + \frac{\SNR}{\nt\nr +  \nt \nr \SNR \epsilon^2_{*,T}} \hat \HRM_{\ell}^{(T)} \hat \HRM_{\ell}^{\dagger (T)} \right) \right] - 1  \right \}  \nonumber \\
 &  \qquad \geq \frac{L - \nt}{L} \left \{ \Expect \left[ \log {\rm det} \left(\IM_{\nr} + \frac{\SNR}{\nt\nr +  \nt \nr \SNR \epsilon^2} \bar \HRM \hermi{\bar \HRM} \right) \right] - 1  \right \}  \label{eq:GMILB_asymptotic_1} \\ 
&  \qquad \geq \frac{L - \nt}{L}  \left ( \Expect \left[ \log {\rm det} \left( \frac{\SNR}{\nt\nr +  \nt \nr\SNR \epsilon^2 } \bar \HRM \hermi{\bar \HRM} \right) \right] - 1  \right ).  \label{eq:GMILB_asymptotic_2}
\end{align}
where the last inequality follows from the lower bound $\log {\rm det} \left( {\sf I} + {\sf A} \right) \geq \log {\rm det} \: {\sf A}$. 
Combining \eqref{eq:GMILB_asymptotic_2} with \eqref{eq:GMILB_Tfixed} yields
\begin{align}
 \gmi (\SNR) & = \lim_{T \to \infty}  \gmi_T (\SNR)   \\
& \geq \frac{L - \nt}{L}  \Big ( \nt \log \SNR -  \nt \log \left( \nt^2  + \nt^2 \SNR \epsilon^2 \right)   +  \Expect \left[ \log {\rm det} \: \bar \HRM  \bar \HRM^\dagger \right]   - 1 \Big  ). \label{eq:gmi-lb-T-infinity}
\end{align}
where in the last inequality we have used the assumption $\nt = \nr $.

\subsubsection{\underline{The Pre-Log}}
We next compute a lower bound on the pre-log by computing the limiting ratio of the RHS of \eqref{eq:gmi-lb-T-infinity} to $\log \SNR$ as $\SNR$ tends to infinity. To this end, we first consider 
\begin{align}
\SNR \: \epsilon^2 &= \SNR \left( 1 - \int^{1/2}_{-1/2} \frac{\SNR \left[ f_H (\lambda) \right]^2}{\SNR f_H (\lambda) + L \nt} d\lambda \right) \\
&= \int^{1/2}_{-1/2} \frac{\SNR f_H(\lambda) L \nt}{\SNR f_H (\lambda) + L \nt} d\lambda.
\end{align}
Since the integrand is bounded by 
\begin{equation}
0 \leq \frac{\SNR f_H(\lambda) L \nt }{\SNR f_H (\lambda) + L} \leq L \nt  \label{eq:bounding-SNR-epsilon-sq}
\end{equation}
it follows that $0 \leq \SNR \: \epsilon^2 \leq L\nt$, which implies that 
\begin{equation}
 \lim_{\SNR \to \infty}  \frac{\log \left( \nt^2 + \nt^2 \SNR \: \epsilon^2 \right)}{\log \SNR} = 0. \label{eq:limitingratio-SNR-epsilon-sq}
\end{equation}
We next consider the term $\Expect \left[ \log {\rm det} \: \bar \HRM  \bar \HRM^\dagger \right] - 1$. Note that by \cite[Lemma A.2]{eurasip:grant:rayleighfading_multi} and by the assumption $\nt = \nr$, we have
\begin{align}
 \Expect\left[ \log {\rm det} ~ \bar \HRM \bar \HRM^\dagger \right] - 1  &=  \nt \log (1 - \epsilon^2) + \sum^{\nt - 1}_{b = 0} \psi ( \nt  - b) - 1 \label{eq:expect-det-H-H}
\end{align}
where $\psi (\cdot)$ is Euler's digamma function \cite{abramowitz:handbook}. Furthermore, since 
\begin{equation}
0 \leq \frac{\SNR \left[ f_H (\lambda) \right]^2}{\SNR f_H (\lambda) + L \nt} \leq f_H (\lambda) 
\end{equation}
we have by the Dominated Convergence Theorem \cite{durrett:probability} that
\begin{equation}
\lim_{\SNR \to \infty} \epsilon^2 = \lim_{\SNR \to \infty} \left( 1 - \int^{1/2}_{-1/2} \frac{\SNR \left[ f_H (\lambda) \right]^2}{\SNR f_H (\lambda) + L \nt} d\lambda \right) = 0 \label{eq:limit-epsilon-SNR-to-infty}
\end{equation}
so $\log (1 - \epsilon^2)$ vanishes as the SNR tends to infinity. Combining \eqref{eq:limit-epsilon-SNR-to-infty} with \eqref{eq:expect-det-H-H} yields 
\begin{equation}
\lim_{\SNR \to \infty} \: \frac{\Expect \left[ \log {\rm det} ~ \bar \HRM \bar \HRM^\dagger \right] - 1}{\log \SNR} = 0. \label{eq:limiting-ratio-det-hatH}
\end{equation}

It thus follows from \eqref{eq:gmi-lb-T-infinity}, \eqref{eq:limitingratio-SNR-epsilon-sq} and \eqref{eq:limiting-ratio-det-hatH} that
\begin{align}
\Pi_{R^*}  & \geq  \nt \left( 1 - \frac{\nt}{L} \right) \\
& =  \min(\nt,\nr) \left( 1 - \frac{\min(\nt,\nr)}{L} \right), \qquad L \leq \frac{1}{2\lambda_D}
\end{align}
where we have used that $\nt = \nr = \min(\nt,\nr)$. Note that the condition $L \leq \frac{1}{2\lambda_D}$ is necessary since otherwise \eqref{eq:epsilon-sq-irrespective-ell-t} would not hold. This proves Theorem \ref{th:pre-log-MIMO}.

\subsection{A Note on Input Distribution}
\label{subsec:note-input-distribution}

The pre-log in Theorem \ref{th:pre-log-MIMO} is derived using codebooks whose entries are drawn i.i.d. from an $\nt$-variate Gaussian distribution with zero mean and identity covariance matrix. However, Gaussian inputs are not necessary to achieve the pre-log \eqref{eq:pre-log_L}. In fact, \eqref{eq:pre-log_L} can be achieved by any i.i.d. inputs having density satisfying $\Expect [  \|\bar \xrv \|^2 ]\leq\nt$ and \eqref{eq:pdf-constraint} and \eqref{eq:peak-density-cons}, namely,
\begin{align}
& p_{\xrv} (\bar \xv)  \leq \frac{K}{\pi^\nt} e^{-\| \bar \xv \|^2},   \quad \bar \xv \in \field^\nt   \label{eq:constraint-input-pdf} \\
& \lim_{\SNR \to \infty} \: \frac{\log K}{\log \SNR} = 0. 
\end{align}
Note that the fact that the inputs have a density implies that $\Expect [ \| \bar \xrv \|^2 ] > 0$. To show that the conditions \eqref{eq:pdf-constraint} and \eqref{eq:peak-density-cons} suffice to achieve \eqref{eq:pre-log_L}, we follow the steps in Section \ref{subsec:achievable-rates} but with $F(\SNR)$ replaced by
\begin{equation}
F(\SNR) = \nr  + \frac{\SNR}{(L - \nt)\nt} \sum^{L-1}_{\ell = \nt}   \Expect \left[\left \| \ERM^{(T)}_\ell \bar \xrv_\ell  \right \|^2_F  \right].
\end{equation}
We then upper-bound $F(\SNR)$ and $\kappa (\theta,\SNR)$ as follows. Using that for any two matrices $\sf A$ and $\sf B$ we have $\| {\sf A} {\sf B} \|^2_F \leq \| {\sf A} \|^2_F \cdot \| {\sf B} \|^2_F$ \cite[Sec. 5.6]{horn:matrixanalysis} and that $\ERM^{(T)}_\ell$ and $\bar \xrv_\ell$ are independent, we can upper-bound $F(\SNR)$ by
\begin{align}
F(\SNR) \leq \nr  + \frac{\SNR}{(L - \nt)\nt}  \sum^{L-1}_{\ell = \nt} \Expect \left[ \left \| \ERM^{(T)}_\ell \right \|^2_F \right] \cdot \Expect \left[ \left \|\bar \xrv_\ell \right \|^2 \right].  \label{eq:F-SNR-UB}
\end{align}
As for $\kappa (\theta,\SNR)$, we have
\begin{align}
 &  \Expect \left[ \left.  {\rm exp} \left(\frac{\theta}{n}   D_k (m')   \right) \right| \yv_k, \hat \HM^{(T)}_{k} \right] \nonumber \\
 &  \qquad = \int_{\bar \xv_k } p_{\xrv} (\bar \xv_k) \: {\rm exp} \left( \frac{\theta}{n} \left \| \yv_k -  \sqrt{ \frac{\SNR }{ \nt} } \: \hat \HM_k^{(T)} \bar \xv_k \right \|^2 \right) d \bar \xv_k \\
 &  \qquad \leq \int_{\bar \xv_k }  \frac{K}{\pi^\nt} \: {\rm exp} \left( - \|\bar \xv_k \|^2 + \frac{\theta}{n} \left \| \yv_k -  \sqrt{ \frac{\SNR }{ \nt} } \: \hat \HM_k^{(T)} \bar \xv_k \right \|^2 \right) d\bar \xv_k \label{eq:inequality-due-to-pdf-constraint}\\
 &  \qquad  =  \frac{K}{{\rm det} \left( \mathsf{I}_\nr - \frac{\theta}{n} \frac{\SNR}{\nt} \hat \HM_{k}^{(T)} \hat \HM_{k}^{\dagger (T)}  \right)} {\rm exp} \left( \frac{\theta}{n} \yv_{k}^\dagger \left( \mathsf{I}_\nr - \frac{\theta}{n} \frac{\SNR}{\nt} \hat \HM_{k}^{(T)} \hat \HM_{k}^{\dagger (T)}  \right)^{-1} \yv_{k}\right). \label{eq:kappa-n-theta-SNR-UB}
\end{align}
Here \eqref{eq:inequality-due-to-pdf-constraint} follows from \eqref{eq:constraint-input-pdf}, and \eqref{eq:kappa-n-theta-SNR-UB} follows by evaluating the integral as in \cite[App. A]{IEEE:weingarten:gaussiancodes}. By following the steps used in Section \ref{subsec:achievable-rates}, and by choosing 
\begin{equation}
\theta = \frac{-1}{\nr +  \SNR \, \nr \epsilon^2_{*,T} \Expect \left[ \| \bar \xrv \|^2 \right]}
\end{equation}
where $\epsilon^2_{*,T}$ is given in \eqref{eq:epsilon-sq-max-T}, we obtain from \eqref{eq:F-SNR-UB} and \eqref{eq:kappa-n-theta-SNR-UB} 
\begin{align}
 I^{\rm gmi}_T  (\SNR)  & \geq    \frac{1}{L}  \sum^{L-1}_{\ell = \nt} \left \{ \Expect \left[ \log {\rm det} \left(\IM_{\nr} +  \frac{\SNR}{\nt \nr + \nt \nr \SNR \epsilon^2_{*,T} \Expect \left[ \| \bar \xrv \|^2 \right] } \hat \HRM_{\ell}^{(T)} \hat \HRM_{\ell}^{\dagger(T)} \right) \right]    \right\} \nonumber \\
& \qquad - \frac{L - \nt}{L} \left(  1 +  \log K  \right).
\end{align}
Taking the limit as $T$ tends to infinity, and repeating the steps used in Section \ref{subsec:achievable-rates} yield
\begin{align}
\gmi (\SNR) & =  \lim_{T \to \infty} I^{\rm gmi}_T  (\SNR)  \\
& \geq  \frac{L - \nt}{L} \left(\Expect \left[ \log\:{\rm det} \left( \frac{\SNR}{\nt \nr + \nt \nr \SNR \: \epsilon^2 \Expect \left[ \| \bar \xrv \|^2 \right] } \bar \HRM \bar \HRM^\dagger \right) \right] - 1 - \log K \right)  \\
& =   \frac{L-\nt}{L} \Big( \nt \log \SNR - \nt \log ( \nt^2 + \nt^2 \SNR \: \epsilon^2  \Expect \left[ \| \bar \xrv \|^2 \right] )  \nonumber \\
&  \quad \qquad \qquad \qquad \qquad \qquad \qquad \qquad   + \: \Expect \left[ \log {\rm det} \: \bar \HRM  \bar \HRM^\dagger \right]   - 1 - \log K \Big) \label{eq:gmi-lb-T-infinity-arbitrary-input}
\end{align}
where we have again used the assumption $\nt = \nr$. We conclude by evaluating the limiting ratio of the RHS of \eqref{eq:gmi-lb-T-infinity-arbitrary-input} to $\log \SNR$ as $\SNR$ tends to infinity. Using \eqref{eq:bounding-SNR-epsilon-sq} and that $ \Expect [ \| \bar \xrv \|^2 ] \leq \nt$ yields 
\begin{equation}
\lim_{\SNR \to \infty}  \frac{\log \left( \nt^2 + \nt^2\SNR \: \epsilon^2  \Expect \left[ \| \bar \xrv \|^2 \right]  \right)}{\log \SNR} = 0. \label{eq:limiting-ratio-snr-eps-sq-gen-input}
\end{equation}
This in turn yields together with \eqref{eq:limiting-ratio-det-hatH}  that 
\begin{equation} 
\lim_{\SNR\to\infty} \frac{I^{\rm gmi}(\SNR)}{\log\SNR} \geq \nt \left(1-  \frac{\nt}{L} \right)
\end{equation} 
provided that
\begin{align}
\lim_{\SNR \to \infty} \: \frac{\log K}{\log \SNR} = 0.  
\end{align}
This concludes the proof.

\section{Proof of Theorem \ref{th:JTD_pre-log}}
\label{sec:proof-mac-pre-log}

In contrast to the proof of Theorem \ref{th:pre-log-MIMO}, for the fading MAC, it is not sufficient to restrict ourselves to the case of $\ntone = \nttwo = \nr$. For example, increasing $\nr$ beyond $\ntone$ and $\nttwo$ does not increase the single-rate pre-logs $\Pi_{R_1^*}$ and $\Pi_{R_2^*}$, but it does increase the pre-log of the achievable sum-rate $\Pi_{R_{1+2}^*}$. For the proof of Theorem \ref{th:JTD_pre-log}, we therefore consider a general setup of $\ntone$, $\nttwo$ and $\nr$.

We derive the achievable pre-logs for the MAC case using a similar approach to the point-to-point case. We first consider the average error probability, averaged over the ensemble of i.i.d. Gaussian codebooks. Let $\bar P_e$ and $\bar P_e (m_1,m_2)$ be the ensemble-average error probability and the ensemble-average error probability corresponding to message $m_1$ and $m_2$ being transmitted, respectively. Due to the symmetry of the codebook construction, $\bar P_e$ is equal to $\bar P_e (1,1)$ and it therefore suffices to consider $\bar P_e (1,1)$ to derive the achievable rates. Let $\mathcal{E}(m_1', m_2')$ denote the event that $D(m_1',m_2') \leq D(1,1)$. Using the union bound, the error probability $\bar P_e (1,1)$ can be upper-bounded as
\begin{align}
&  \bar P_e (1,1) \nonumber \\
& =   \Pr \left \{ \bigcup_{(m_1', m_2') \neq (1,1)}  \mathcal{E}(m_1', m_2') \right \} \\
& \leq  \Pr \left \{ \bigcup_{ m_1'\neq 1}  \mathcal{E}(m_1', 1) \right \} + \Pr \left \{ \bigcup_{ m_2'\neq 1}  \mathcal{E}(1,m_2') \right \} + \Pr \left \{ \bigcup_{ m_1'\neq 1} \bigcup_{ m_2'\neq 1} \mathcal{E}(m_1', m_2') \right \}.  \label{eq:rcub-probs-MAC}
\end{align}
We next analyze these probabilities corresponding to the error events $(m_1' \neq 1, m_2' = 1)$, $(m_1' = 1, m_2' \neq 1)$ and $(m_1' \neq 1, m_2' \neq 1)$. Let the matrix $\mathbb{E}_{s,k}^{(T)}$, $s=1,2$ with entries $E_{s,k}^{(T)} (r,t)$ be the estimation-error matrix in estimating $\mathbb{H}_{s,k}$, i.e.,
\begin{equation}
\mathbb{E}_{s,k}^{(T)} = \mathbb{H}_{s,k} - \hat{\mathbb{H}}_{s,k}^{(T)}.
\end{equation}
To facilitate the analysis, we first generalize $F(\SNR)$ and $\mathcal{T}_\delta$ in the point-to-point case (cf. \eqref{eq:F-SNR-defined} and \eqref{eq:typical_set_delta}) to the MAC case, i.e., 
\begin{align}
F(\SNR) & = \nr  + \frac{\SNR}{L - \ntone -\nttwo}  \sum^{L-1}_{\ell= \ntone + \nttwo} \Expect \left[ \left \| \ERM^{(T)}_{1,\ell} \right \|^2_F +  \left \|  \ERM^{(T)}_{2,\ell} \right \|^2_F  \right],    \label{eq:expectation-F-mac} \\
\mathcal{T}_{\delta} &  = \Bigg \{  \left( \xv_{s,k}, \yv_k, \hat \HM_{s,k}^{(T)} \right), k=0,\dotsc, n'-1, s = 1,2 : \nonumber \\
 &   \qquad \: \: \left | \frac{1}{n} \sum_{k \in  \mathcal{D}^{(n')}} \left
 \|\yv_k - \sqrt{\SNR} ~ \hat \HM^{(T)}_{1,k} \xv_{1,k} -  \sqrt{\SNR} ~ \hat \HM^{(T)}_{2,k} \xv_{2,k} \right\|^2  - F(\SNR) \right | < \delta  \Bigg \} 
\end{align}
for some $\delta > 0$, with $n'$ given in \eqref{eq:total_length_mac} and $\mathcal{D}^{(n')} = \{0,\dotsc,n' - 1\} \cap \mathcal{D}$. Using $F(\SNR)$ and the typical set $\mathcal{T}_\delta$, we continue by evaluating the GMI for each of the three probabilities on the RHS of \eqref{eq:rcub-probs-MAC}  corresponding to the error events $(m_1' \neq 1, m_2' = 1)$, $(m_1' = 1, m_2' \neq 1)$ and $(m_1' \neq 1, m_2' \neq 1)$.

\subsubsection{\underline{Error Event $(m_1' \neq 1, m_2' = 1)$}}

Following the steps as used in Section~\ref{subsec:achievable-rates} to derive \eqref{eq:bounding-typical-atypical}, we can upper-bound the ensemble-average error probability for the error event $\mathcal{E} (m'_1, 1)$, $m'_1 \neq 	1$  using $\mathcal{T}_\delta$ and its complement $\mathcal{T}^c_\delta$ as 
\begin{align}
&  \Pr \left \{ \bigcup_{ m_1'\neq 1}  \mathcal{E}(m_1', 1) \right \} \nonumber \\
&  \quad \leq e^{nR_1} \cdot \Pr \left \{ \left. \frac{1}{n} \cdot D(m'_1, 1) < F(\SNR) + \delta  \,\right |  \left \{  \left( \xrv_{s}^{n'}(1), \yrv^{n'}, \hat \HRM_{s}^{(T), n'} \right),  s=1,2 \right \}  \in \mathcal{T}_{\delta}    \right \} \nonumber \\
&  \qquad + \Pr \left \{   \left \{ \left( \xrv_{s}^{n'}(1), \yrv^{n'}, \hat \HRM_{s}^{(T), n'} \right),  s=1,2 \right \}  \in \mathcal{T}_{\delta}^c  \right \}, \quad m_1' \neq 1. \label{eq:rcub-error-user1-mac}
\end{align}
Note that the second probability on the RHS of \eqref{eq:rcub-error-user1-mac} vanishes as $n$ tends to infinity, which can be shown along the lines of the proof of Lemma \ref{lemma:F-SNR-typical-set}. 

The GMI for User 1 gives the rate of exponential decay of the term 
\begin{equation}
 \Pr \left \{ \left. \frac{1}{n} \cdot D(m'_1, 1) < F(\SNR) + \delta  \right |  \left \{  \left( \xrv_{s}^{n'}(1), \yrv^{n'}, \hat \HRM_{s}^{(T), n'} \right),  s=1,2 \right \}  \in \mathcal{T}_{\delta}    \right \} 
\end{equation}
as $n \to \infty$. The evaluation of the GMI for User 1 requires the expression of the log moment-generating function of the metric $D(m'_1, 1)$ associated with an incorrect message $m'_1 \neq 1$---conditioned on the channel outputs, on $m_2' = 1$, and on the fading estimates---which we shall denote by $\kappa_{1,n} (\theta, \yv^{n'}, \xv^{n'}_2 (1), \hat \HM^{(T),n'}_1, \hat \HM^{(T),n'}_2)$, i.e.,
\begin{align}
&\kappa_{1,n} \left(\theta, \yv^{n'}, \xv^{n'}_2 (1), \hat \HM^{(T),n'}_1, \hat \HM^{(T),n'}_2 \right) \nonumber \\
 &= \log \Expect \bigg[   {\rm exp} \left(\frac{\theta}{n} \sum_{k \in \mathcal{D}^{(n')}}  D_k (m_1', 1)   \right) \bigg| \left \{ \left(\yv_k, \xv_{2,k} (1), \hat \HM^{(T)}_{1,k}, \hat \HM^{(T)}_{2,k}\right),\: k \in \mathcal{D}^{(n')}   \right\} \bigg]  
\end{align}
where we have defined
\begin{equation}
D_k (m_1',m_2')  \triangleq  \left \| \yv_k - \sqrt{\SNR} ~ \hat \HM^{(T)}_{1,k} \xv_{1,k} {(m_1')}  - \sqrt{\SNR} ~ \hat \HM^{(T)}_{2,k} \xv_{2,k} {(m_2')} \right \|^2. \label{eq:Dk-m1-m2}
\end{equation}
Following the steps used in Section~\ref{subsec:achievable-rates} to obtain \eqref{eq:conditional-log-mgf-n-to-sum-mgf} and \eqref{eq:sum-conditional-log-mgf}, it can be shown that
\begin{align}
& \kappa_{1,n} \left(\theta, \yv^{n'}, \xv^{n'}_2 (1), \hat \HM^{(T),n'}_1, \hat \HM^{(T),n'}_2 \right) = \nonumber \\
&  \sum_{k \in \mathcal{D}^{(n')}} \bigg \{ \frac{\theta}{n} \left( \yv_k - \sqrt{\SNR}\, \hat \HM_{2,k}^{(T)} \xv_{2,k} (1) \right)^\dagger   \left( \mathsf{I}_\nr - \frac{\theta}{n} \SNR \,\hat \HM_{1,k}^{(T)} \hat \HM_{1,k}^{\dagger(T)}  \right)^{-1} \left( \yv_k - \sqrt{\SNR} \,\hat \HM_{2,k}^{(T)} \xv_{2,k} (1) \right)\nonumber\\
 &  \qquad\qquad \qquad\qquad \qquad\qquad \qquad\qquad {} - \log {\rm det} \left(\mathsf{I}_\nr - \frac{\theta}{n}  \SNR\, \hat \HM_{1,k}^{(T)} \hat \HM_{1,k}^{\dagger(T)} \right ) \bigg \}. 
 \end{align}
Then, following the steps used in Section~\eqref{subsec:achievable-rates} to derive \eqref{eq:log-mgf-n-to-infinity}--\eqref{eq:converge-log-mgf}, we have that for all $\theta < 0$
\begin{align}
&  \lim_{n \to \infty} \frac{1}{n} \cdot \kappa_{1,n} \left(n\theta, \yv^{n'}, \xv^{n'}_2 (1), \hat \HM^{(T),n'}_1, \hat \HM^{(T),n'}_2 \right)  \nonumber \\
&  \ = \frac{1}{L - \ntone - \nttwo } \sum^{L- 1}_{\ell= \ntone + \nttwo} \left( g_{1,\ell} (\theta,\SNR)  -  \Expect \left[ \log {\rm det} \left(\mathsf{I}_\nr - \theta\, \SNR\, \hat \HRM_{1,\ell}^{(T)} \hat \HRM_{1,\ell}^{\dagger(T)} \right)\right] \right) \label{eq:kappa-1-n-theta-SNR-n-to-infty} \\
& \triangleq \kappa_1 (\theta, \SNR) \nonumber \label{eq:def:kappa1}
\end{align}
almost surely, where \eqref{eq:kappa-1-n-theta-SNR-n-to-infty} should be regarded as the definition of $\kappa_1 (\theta,\SNR)$. Here we define
\begin{align}
 g_{1,\ell} (\theta,\SNR) & \triangleq \Expect \bigg[ \theta \left( \yrv_\ell - \sqrt{\SNR}\, \hat \HRM_{2,\ell}^{(T)} \xrv_{2,\ell} \right)^\dagger  \times\left( \mathsf{I}_\nr - \theta \,\SNR \,\hat \HRM_{1,\ell}^{(T)} \hat \HRM_{1,\ell}^{\dagger(T)}  \right)^{-1}\nonumber\\
 &  \qquad\qquad \qquad\qquad \qquad\qquad \qquad\qquad {} \times \left( \yrv_\ell - \sqrt{\SNR} \,\hat \HRM_{2,\ell}^{(T)} \xrv_{2,\ell} \right) \bigg].
\end{align}

Following the derivation in \cite{IEEE:lapidoth:fadingchannels_howperfect,IEEE:weingarten:gaussiancodes}, we can then upper-bound the ensemble-average error probability ($\mathcal{E} (m'_1, 1)$, $m'_1 \neq 1$) for any $\delta' > 0$ as
\begin{equation}
\Pr \left \{ \bigcup_{ m_1'\neq 1}  \mathcal{E}(m_1', 1) \right \} \leq {\rm exp} \left( n R_1 \right) {\rm exp} \left(- n \left( \gmi_{1,T} (\SNR)  - \delta' \right) \right) + \varepsilon_1 ( \delta', n) \label{eq:rcub-gmi-mac-1}
\end{equation}
for some $\varepsilon_1 (\delta', n)$ satisfying
\begin{equation}
\lim_{n \to \infty} \varepsilon_1 ( \delta', n)  = 0,\quad \delta' > 0.
\end{equation}
On the RHS of \eqref{eq:rcub-gmi-mac-1}, $I^{\rm gmi}_{1,T} (\SNR)$ denotes the GMI for User 1 as a function of $\SNR$ for a fixed $T$ and is given by
\begin{equation}
I^{\rm gmi}_{1,T} (\SNR) =  \frac{L - \ntone - \nttwo}{L} \left( \sup_{\theta < 0} ~ \big(\theta F(\SNR) - \kappa_1 (\theta, \SNR)  \big)  \right).\label{eq:gmi_user_1_def}
\end{equation}
The pre-factor $(L-\ntone - \nttwo)/L$ equals the fraction of time used for data transmission. The bound \eqref{eq:rcub-gmi-mac-1} implies that for all rates below $\gmi_{1,T} (\SNR)$, decoding the message from User 1 using the scheme described in Section \ref{sec:mac} has vanishing error probability as $n$ tends to infinity. Combining \eqref{eq:expectation-F-mac} and \eqref{eq:kappa-1-n-theta-SNR-n-to-infty} with \eqref{eq:gmi_user_1_def} yields
\begin{align}
\gmi_{1,T} (\SNR) = \sup_{\theta < 0}~ \frac{1}{L} \sum^{L-1}_{\ell = \ntone + \nttwo}  \bigg \{ \theta & \left( \nr + \SNR \: \Expect \left[ \left \| \ERM^{(T)}_{1,\ell} \right \|^2_F +  \left \|  \ERM^{(T)}_{2,\ell} \right \|^2_F   \right] \right) - g_{1,\ell} (\theta,\SNR)  \nonumber \\
& \qquad \qquad +   \Expect \left[ \log {\rm det} \left(\mathsf{I}_\nr - \theta\, \SNR\, \hat \HRM_{1,\ell}^{(T)} \hat \HRM_{1,\ell}^{\dagger(T)} \right)\right]  \bigg \}. \label{eq:gmi-1-T-sup}
\end{align}

As the supremum \eqref{eq:gmi-1-T-sup} is difficult to evaluate, we next consider a lower bound on $\gmi_{1,T}(\SNR)$. By noting $g_{1,\ell} (\theta,\SNR) \leq 0$ for $\theta \leq 0$ (which can be shown using the technique developed in \cite[App. D]{IEEE_IT_Asyhari_Nearest_neighbour}) and by choosing\footnote{As pointed in Section \ref{sec:proof-theorem-1}, this choice of $\theta$ yields a good lower bound at high SNR.} 
\begin{equation}
 \theta = \frac{-1}{\nr +  \nr \left( \ntone + \nttwo \right)  \SNR \,\epsilon^2_{*,T} }
\end{equation}
where
\begin{equation}
\label{eq:estimationerror}
\epsilon^2_{*,T} = \max_{\substack{s = 1,2,\\ r = 1,\dotsc,\nr,\\ t = 1,\dotsc,\nts,\\ \ell = \ntone+\nttwo,\dotsc,L-1}} ~ \Expect \left[ \left|E_{s,\ell}^{(T)} (r,t) \right|^2 \right]
\end{equation}
 we obtain a lower bound on $\gmi_{1,T}(\SNR)$
\begin{align}
\gmi_{1,T} (\SNR) \geq  \frac{1}{L} \sum^{L-1}_{\ell= \ntone + \nttwo}\Expect \Bigg[ \log {\rm det} \Bigg(\mathsf{I}_\nr + \frac{\SNR \: \hat \HRM_{1,\ell}^{(T)} \hat \HRM_{1,\ell}^{\dagger (T)}}{\nr +   \nr \left(\ntone + \nttwo \right)  \SNR \,\epsilon^2_{*,T}}  \Bigg)  - 1\Bigg]. 
 \label{eq:lower_bound_gmi_1_Tfixed} 
\end{align}

We continue by analyzing the RHS of \eqref{eq:lower_bound_gmi_1_Tfixed} in the limit as the observation window $T$ of the channel estimator tends to infinity. To this end, we note that, for $L\leq\frac{1}{2\lambda_D}$, the variance of the interpolation error $\Expect [ |E_{s,\ell}^{(T)} (r,t) |^2 ]$  tends to \eqref{eq:fadingestimate_error} (with $\SNR $ in \eqref{eq:fadingestimate_error} replaced by $\nt \SNR$),\footnote{Note the difference between the point-to-point channel model \eqref{eq:channelmodel} and the MAC channel model \eqref{eq:channel}.} so
\begin{equation}
\lim_{T \rightarrow \infty} \Expect \left[ \left|E_{s,\ell}^{(T)} (r,t) \right|^2 \right] =  \epsilon^2 = 1 - \int^{1/2}_{-1/2} \frac{\SNR \left[f_H (\lambda) \right]^2}{\SNR f_H (\lambda) + L} d\lambda \label{eq:interpolation-error-mac}
\end{equation}
irrespective of $s,\ell,r$ and $t$. Hence, irrespective of $\ell$, the estimate $\hat \HRM_{1,\ell}^{(T)}$ tends to $\bar \HRM_1$ in distribution as $T$ tends to infinity, so
\begin{align}
\frac{\hat \HRM_{1,\ell}^{(T)} \hat \HRM_{1,\ell}^{\dagger (T)}}{ \nr +   \nr \left( \ntone  + \nttwo \right) \SNR \,\epsilon^2_{*,T}} \quad \stackrel{d}{\longrightarrow}  \quad \frac{\bar \HRM_1 \bar \HRM^\dagger_1}{ \nr + \nr \left( \ntone + \nttwo \right) \SNR\, \epsilon^2} 
\end{align}
where the $\nr \times \ntone$ entries of $\bar \HRM_1$ are i.i.d., circularly-symmetric, complex-Gaussian random variables with zero mean and variance $(1 - \epsilon^2)$. Using Portmanteau's lemma (as used in \eqref{eq:GMILB_asymptotic_1}), we obtain  that
\begin{align}
 \gmi_1(\SNR) & = \lim_{T\to\infty}\gmi_{1,T} (\SNR)   \\
&\geq \frac{L- \ntone  - \nttwo}{L} \Bigg( \Expect \left[ \log {\rm det} \left(\mathsf{I}_\nr + \frac{\SNR\, \bar \HRM_1 \bar \HRM_1^\dagger }{ \nr +   \nr \left( \ntone + \nttwo \right) \SNR \,\epsilon^2 } \right) \right] - 1\Bigg)   \label{eq:portmanteau-mac-1} \\
&\geq  \frac{L -  \ntone - \nttwo}{L}\min(\nr,\ntone)\Big[\log\SNR  -  \log\big(\nr +  \nr ( \ntone + \nttwo ) \,\SNR \,\epsilon^2\big) \Big] \nonumber\\
& \quad {} + \frac{L -  \ntone - \nttwo}{L}\Psi_1  \label{eq:lower_bound_gmi_1_asymptotic_ntonesmall}
\end{align}
where
\begin{equation}
 \Psi_1 \triangleq \begin{cases}
         \Expect \left[ \log {\rm det} \: \bar \HRM_1  \bar \HRM^\dagger_1 \right]   - 1, &  $\nr \leq \ntone$\\
         \Expect \left[ \log {\rm det} \: \bar \HRM^\dagger_1 \bar \HRM_1 \right] - 1, &  $\nr > \ntone$.
        \end{cases}
\end{equation}
Here the last inequality follows by lower-bounding $\log {\rm det} \left( {\sf I} + {\sf A} \right) \geq \log {\rm det} {\sf A}$.

By evaluating the RHS of \eqref{eq:lower_bound_gmi_1_asymptotic_ntonesmall} in the same way as evaluating the RHS of \eqref{eq:gmi-lb-T-infinity} in Section \ref{subsec:achievable-rates}, we obtain a lower bound for the maximum achievable pre-log for User 1 as 
\begin{equation}
\Pi_{R^*_1}  \geq \min(\nr,\ntone)\left(1-\frac{\ntone + \nttwo }{L}\right), \quad L\leq\frac{1}{2\lambda_D}.\label{eq:firstbound}
\end{equation}
Here instead of assuming $\nt = \ntone + \nttwo = \nr$, we have used a general setup of $\ntone$, $\nttwo$ and $\nr$. Note that the condition $L\leq1/(2\lambda_D)$ is necessary since otherwise \eqref{eq:fadingestimate_error} would not hold. This yields one boundary of the pre-log region presented in Theorem \ref{th:JTD_pre-log}.

\vspace{4mm}

\subsubsection{\underline{Error Event $(m_1'=1, m_2' \neq 1)$}}

This follows from the proof for the error event $(m'_1\neq 1, m'_2=1)$ by swapping User 1 with User 2. We thus have
\begin{align}
\Pi_{R^*_2} &  \geq   \min(\nr,\nttwo)\left(1-\frac{\ntone + \nttwo }{L}\right), \quad L\leq\frac{1}{2\lambda_D}\label{eq:secondbound}
\end{align}
yielding the second boundary of the pre-log region presented in Theorem~\ref{th:JTD_pre-log}.

\vspace{4mm}

\subsubsection{\underline{Error Event $(m_1' \neq 1, m_2' \neq 1)$}}

The analysis on the achievable sum rate corresponding to the joint error event $\mathcal{E}(m'_1, m'_2)$, $(m_1' \neq 1, m_2' \neq 1)$ in the MAC case follows the same analysis as in the point-to-point case (Section \ref{subsec:achievable-rates}). More specifically, the sum of the GMI $\gmi_{1+2,T} (\SNR)$ can be viewed as the GMI of an $\nr \times (\ntone + \nttwo)$-dimensional point-to-point MIMO channel with fading matrix at time $k$, $[\HRM_{1,k}, \HRM_{2,k}]$, and fading estimate matrix at time $k$, $\left[ \hat \HRM_{1,k}^{(T)}, \hat \HRM_{2,k}^{(T)} \right]$. The maximum achievable sum-rate pre-log can therefore be obtained using the same approaches as in Section \ref{subsec:achievable-rates}, but with arbitrary $\nr$ and $\nt = \ntone + \nttwo$. It can be shown that the maximum achievable sum-rate pre-log $\Pi_{R_{1+2}^*}$ is lower-bounded by
\begin{align}
 \Pi_{R_{1+2}^*} \geq \min \left( \nr, \ntone + \nttwo \right)   \left(1 - \frac{\ntone + \nttwo}{L} \right), \quad L \leq \frac{1}{2 \lambda_D}.\label{eq:thirdbound}
\end{align}
On the RHS of \eqref{eq:thirdbound}, the term $ \min \left( \nr, \ntone + \nttwo \right)$ corresponds to the MIMO gain, which is given by the minimum number of receive and transmit antennas, and the term $\left(1 - \frac{\ntone + \nttwo}{L} \right)$ corresponds to the fraction of time for data transmission, which changes for arbitrary number of transmit antennas in comparison to the proof of the point-to-point channel. This yields the third boundary of the pre-log region presented in Theorem~\ref{th:JTD_pre-log}.

\section{Conclusion}
\label{sec:conclusion} 

In this paper we studied a communication scheme for MIMO fading channels that estimates the fading via transmission of pilot symbols at regular intervals, and feeds the fading estimates to the nearest neighbor decoder. Restricting ourselves to fading processes with a bandlimited power spectral density, we studied the information rates achievable with this scheme at high SNR. Specifically, we analyzed the achievable rate pre-log, defined as the limiting ratio of the achievable rate to the logarithm of the SNR in the limit as the SNR tends to infinity. 

We showed that, in order to obtain fading estimates whose variance vanishes as the SNR tends to infinity, the portion of time required for pilot transmission must be greater or equal to the number of transmit antennas times twice the bandwidth of the fading power spectral density. We demonstrated that, in this case, the nearest neighbor decoder achieves the capacity pre-log of the coherent fading channel times the fraction of time used for the transmission of data. Hence, the loss with respect to the coherent case is solely due to the transmission of pilots used to obtain accurate fading estimates. Our achievability bounds are tight in the sense that any scheme using as many pilots as our proposed scheme cannot achieve a higher pre-log using a nearest neighbor decoder. Furthermore, if the inverse of twice the bandwidth of the fading process is an integer, then, for MISO channels, our scheme achieves the capacity pre-log of the non-coherent fading channel derived by Koch and Lapidoth \cite{IEEE:koch:fadingnumber_degreeoffreedom}. For non-coherent MIMO channels, our scheme achieves the best so far known lower bound on the capacity pre-log obtained by Etkin and Tse \cite{IEEE:etkin:degreeofffreedomMIMO}. Since the last result only yields a lower bound on the capacity pre-log of MIMO channels, there may exist other schemes achieving a better pre-log than our scheme.

\appendices

\section{Proof of Lemma \ref{lemma:interpolator-property-fix-T}}
\label{sec:proof-lemma:interpolator-property-fix-T}

\begin{enumerate}

\item  By the orthogonality principle \cite{poor-detect-estimation}, we have that $\hat H_{k}^{(T)} (r,t)$ and $E_{k}^{(T)} (r,t)$ are uncorrelated. Noting that the pilot symbols are unity, we can write \eqref{eq:fading-estimate-pilot-only-observation} as 
\begin{equation}
\hat H_k^{(T)} (r,t) = \sum^{k + T L}_{\substack{ k' = k - TL:\\ k' \in \mathcal{P}}} a_{k'} (r,t) \left( \sqrt{\frac{\SNR}{\nt}} H_{k'} (r,t) + Z_{k'} (r) \right), \qquad k \in \mathcal{D}. \label{eq:fading-estimate-pilot-expanded}
\end{equation} 
Since the processes $\{ H_k (r,t),\: k \in \integ \}$ and $\{ Z_k (r),\: k \in \integ \}$ are zero-mean complex-Gaussian processes,  we have from \eqref{eq:fading-estimate-pilot-expanded} and the orthogonality principle that $\hat H_{k}^{(T)} (r,t)$ and $E_{k}^{(T)} (r,t)$ are independent zero-mean complex-Gaussian random variables.

\item Recall from Section \ref{subsec:channelestimator} that the time index $k$ can be written as $k = jL + \ell$. Then, for $k \in \mathcal{D}$, we have $\ell = \nt,\dotsc,L-1$, and for $k \in \mathcal{P}$ we have $\ell = 0,\dotsc,\nt-1$. Since the pilot vectors are transmitted sequentially from ${\bm p}_1$ to ${\bm p}_\nt$, we have for $k \in \mathcal{P}$ that  
\begin{equation}
{\xv}_{jL + \ell } = {\bm p}_{\ell +1}, \qquad \ell =0,\dotsc,\nt-1
\end{equation}
namely the $(\ell + 1)$-th pilot vector, $\ell =0,\dotsc,\nt-1$ is used to estimate the fading coefficients from transmit antenna $t$. We next note that, in order to estimate $H_k(r,t)$, there is no loss in optimality by considering only the outputs $Y_{k'} (r)$ for $k' \in \mathcal{P} \cap \{k-TL,\dotsc,k+TL \}$ satisfying
\begin{equation}
k' \: {\rm mod} \: L = t - 1.
\end{equation}
Indeed, the channel outputs $Y_{k'} (r)$, $k'\: {\rm mod} \: L \neq t-1$ correspond to $H_{k'} (r,t')$,  $t' \neq t$, which are independent from $H_k (r,t)$ since we have assumed that the fading processes corresponding to different transmit and receive antennas are independent. It follows that for the estimation at $k=jL + \ell$, the coefficients $a_{k'}  (r,t)$ that minimize the mean-squared error depend only on $L$ and $\ell$ \cite{IEEE:ohno:averageratePSAM}. The fading estimate \eqref{eq:fading-estimate-pilot-only-observation} can then be expressed  as 
\begin{align}
 \hat H_{jL + \ell}^{(T)} (r,t) & = \sum^{ T - 1 }_{\tau = -T } \alpha_{-\tau L, \ell}  (r,t) Y_{(j-\tau) L + t - 1} (r) \\ 
 & = \sum^{T-1}_{\tau = -T} \alpha_{-\tau L, \ell} (r,t)  \left(  \sqrt{\frac{\SNR}{\nt}}   H_{(j-\tau)L + t - 1}  (r,t) + Z_{(j-\tau)L + t-1} (r) \right) \label{eq:fad-estimate-jL-ell} 
\end{align}
where  for a given $L$ and  $\ell = \nt, \dotsc, L-1$, we have defined
\begin{equation}
\alpha_{-\tau L, \ell} (r,t) \triangleq a_{{(j-\tau) L +  t - 1} }(r,t), \quad \tau = -T,\dotsc,T-1.
\end{equation}

Noting again that the $\nr \cdot \nt$ processes $\{ H_k (r,t),\: k \in \integ \}$ are independent from each other and have the same law, we obtain the following results from  \eqref{eq:fad-estimate-jL-ell}.

\begin{enumerate}
\item For a given $t$, the time differences between the index of interest---($jL + \ell$)---and the positions of  pilots---($(j-\tau)L+t-1$)---do not depend on $r$. It thus follows that for a given $t$, the optimal coefficients $\alpha_{-\tau L, \ell} (r,t)$ are identical for all $r=1,\dotsc,\nr$ \cite{IEEE:ohno:averageratePSAM}. This implies that for a given $t$ and $\ell$, the $\nr$ processes 
$$\{ (\hat H_{jL+\ell}^{(T)} (1,t), E_{jL+\ell}^{(T)} (1,t)), \: j \in \integ \},\dotsc,\{ (\hat H_{jL+\ell}^{(T)} (\nr,t), E_{jL+\ell}^{(T)} (\nr,t)), \: j \in \integ \}$$
are independent and have the same law. 

 \item For a given $r$, the time differences between the index of interest---($jL + \ell$)---and the position of  pilots---($(j-\tau)L+t-1$, $\tau=-T,\dotsc,T-1$)---depend on $t$.  It thus follows from \cite{IEEE:ohno:averageratePSAM} that for a given $r$, the optimal coefficients $\alpha_{-\tau L, \ell} (r,t)$ are generally different for $t=1,\dotsc,\nt$. This implies that for a given $r$ and $\ell$, the $\nt$ processes 
 $$\{ (\hat H_{jL+\ell}^{(T)} (r,1), E_{jL+\ell}^{(T)} (r,1)), \: j \in \integ \},\dotsc,\{ (\hat H_{jL+\ell}^{(T)} (r,\nt), E_{jL+\ell}^{(T)} (r,\nt)), \: j \in \integ \}$$
 are independent but have  different laws. 
\end{enumerate}

\item  We first note that $\{ \HRM_k,\: k \in \integ \}$ is an ergodic Gaussian process, which implies that it is also a weakly mixing process \cite{IEEE:Sethuraman:capaiy-per-unit-energy}. (See \cite{brown_ergodic} for a definition of a weakly-mixing process.) Since $\{\zrv_k,\: k \in \integ \}$ is an i.i.d. Gaussian process and independent from $\{ \HRM_k,\: k \in \integ \}$, it follows from \cite[Prop. 1.6]{brown_ergodic} that $\{ (\HRM_k, \zrv_k),\: k \in \integ \}$ is jointly ergodic.

We next evaluate the process $\{ (\hat \HRM_k^{(T)}, \HRM_k, \zrv_k), \: k \in \mathcal{D} \}$. Note that this process cannot be expressed directly as a time-invariant function of $\{ (\HRM_k, \zrv_k ), \: k \in \integ \}$. Indeed, by assuming $k = jL + \ell$, we can see from \eqref{eq:fad-estimate-jL-ell} that the function to produce $\hat \HRM_k^{(T)}$ from $\{ (\HRM_k, \zrv_k ), \: k \in \integ \}$ depends on the time index $k$ via $\ell$, for $\ell = \nt,\dotsc,L-1$ (corresponding to time indices for data transmission). As such, to facilitate the analysis, we need to introduce a ``dummy'' matrix-valued process $\{ \ARM_{k, \ell}, \: k \in \integ \}$ where $\ARM_{k,\ell}$ has $\nr \times \nt$ entries, and where its entry at row $r$ and column $t$ is given by
\begin{equation}
A_{k,\ell} (r,t) = \sum^{T-1}_{\tau = -T} \alpha_{-\tau L, \ell} (r,t)  \left(  \sqrt{\frac{\SNR}{\nt}}   H_{k-\tau L - \ell + t - 1}  (r,t) + Z_{k-\tau L - \ell + t - 1} (r) \right).  \label{eq:def-A-k-ell-r-t}
\end{equation} 
Here the coefficients $\alpha_{-\tau L, \ell}$, $\tau = -T,\dotsc,T,$ have the same value as those in \eqref{eq:fad-estimate-jL-ell} for a given $L$ and $\ell$. Consequently, $\{ \ARM_{k, \ell}, \: k \in \integ \}$ is a time-invariant function of $\{ (\HRM_k, \zrv_k),\: k \in \integ \}$ that coincides with $\hat\HRM_k^{(T)}$ for $k=jL+\ell$. This in turn implies that $\{ ({\mathbb A}_{k,\ell}, \HRM_k, \zrv_k), \: k \in \integ \}$ is jointly weakly mixing. Furthermore, by the definition of weakly mixing \cite{Petersen_Ergodic,brown_ergodic,IEEE:Sethuraman:capaiy-per-unit-energy}, the process $\{ ({\mathbb A}_{jL + \ell, \ell}, \HRM_{jL + \ell}, \zrv_{jL + \ell}), \: j \in \integ \}$ for any $\ell = 0,\dotsc,L-1$ is also jointly weakly mixing. Since for $k = jL + \ell$, $k \in \mathcal{D}$, the matrix $\ARM_{jL + \ell,\ell}$ is identical to $\hat \HRM^{(T)}_{jL + \ell}$, it follows that the process  $\{(\hat \HRM_{jL + \ell}^{(T)}, \HRM_{jL + \ell},  \zrv_{jL+\ell} ),\: j \in \integ\}$ for each $\ell=\nt,\dotsc,L-1$ is jointly weakly mixing, which implies ergodicity. 

We finally evaluate the joint behavior of the two processes $\{(\hat \HRM_{jL + \ell}^{(T)}, \HRM_{jL + \ell},  \zrv_{jL+\ell} ),\: j \in \integ\}$ and $\{ \xrv_{jL + \ell},\: j \in \integ \}$ for $\ell \in \{ \nt,\dotsc,L-1 \}$. Since $\{ \xrv_{jL + \ell},\: j \in \integ \}$ for $\ell \in \{ \nt, \dotsc,L-1 \}$ is i.i.d. and independent from  $\{(\hat \HRM_{jL + \ell}^{(T)}, \HRM_{jL + \ell},  \zrv_{jL+\ell} ),\: j \in \integ\}$, we have by \cite[Lemma 2]{IEEE_IT:YH-kim:coding-stationary-fb} that the process 
$$ \{ (\hat \HRM^{(T)}_{jL+\ell}, \: \HRM_{jL + \ell},   \: \zrv_{jL + \ell}, \:  \xrv_{jL + \ell}), \: j \in \integ  \}, \quad \ell \in \{\nt, \dotsc,L-1 \}$$ 
is jointly ergodic. This proves Part 3) of Lemma \ref{lemma:interpolator-property-fix-T}.

\item Note that the process $\{ \hat \HRM^{(T)}_k,\: k \in \mathcal{D}\}$ is a function of $\{(\HRM_k, \zrv_k), k \in \mathcal{P} \}$. Since $\{ \zrv_k,\: k \in \mathcal{D} \}$ has zero mean and is independent from $\{ (\HRM_k, \zrv_k),\: k \in \mathcal{P} \}$ and  $\{ \xrv_k,\: k \in \mathcal{D} \}$, it follows that for any  of $\ell = \nt,\dotsc,L-1$ (which correspond to $k \in \mathcal{D}$) 
\begin{equation}
\Expect \left[  \zrv_\ell^\dagger \hat \HRM^{(T)}_\ell \xrv_\ell  \right] = 0. \label{eq:trace-inner-outer-expec}
\end{equation}

\end{enumerate}

\section{Variance of the Interpolation Error for $L > \frac{1}{2\lambda_D}$ }
\label{sec:proof-lemma:interpolator-property-T-infinity}

Recall that, as $T$ tends to infinity, we have that, irrespective of $j$ and $r$, the variance of the interpolation error \eqref{eq:interpolation-error-variance-T-infinity-general}, namely
\begin{align}
\epsilon^2_\ell (t) =  1 - \int^{1/2}_{-1/2} \frac{\SNR \left|f_{L,\ell - t + 1}(\lambda) \right|^2}{\SNR f_{L,0} (\lambda) + \nt} d \lambda \label{eq:proof-error-asymptotic-T}
\end{align}
where 
\begin{equation}
 f_{L,\ell} (\lambda) = \frac{1}{L} \sum^{L-1}_{\nu=0} \bar{f}_H \left( \frac{\lambda - \nu}{L} \right)e^{\ii 2\pi \ell \frac{\lambda - \nu}{L} }, \qquad  -\frac{1}{2} \leq \lambda \leq \frac{1}{2}. \label{eq:undersamp_fHell_general-proof}
\end{equation}
In order to analyze the behavior of $\epsilon^2_\ell (t)$ for $L > \frac{1}{2\lambda_D}$, we first express $L$ as
\begin{equation}
L = \frac{1}{2\lambda_D} + \varepsilon
\end{equation}
for some $\varepsilon > 0$. The variance of the interpolation error \eqref{eq:proof-error-asymptotic-T} can be lower-bounded as\begin{align}
 \epsilon^2_\ell (t)
 =  & \int^{1/2}_{-1/2} \frac{\nt f_{L,0} (\lambda)}{\SNR f_{L,0} (\lambda) + \nt} d\lambda  \nonumber \\
& + \int^{1/2}_{-1/2} \frac{\SNR \left( \left[  f_{L, 0} (\lambda) \right]^2 - \left|f_{L, \ell-t+1} (\lambda) \right|^2 \right)}{\SNR f_{L,0} (\lambda) + \nt} d\lambda  \label{eq:decomposition-integral-est-error-aliasing} \\
\geq &  \int^{1/2}_{-1/2} \frac{\SNR \left( \left[ f_{L, 0} (\lambda) \right]^2 - \left|f_{L, \ell-t+1} (\lambda) \right|^2 \right)}{\SNR f_{L,0} (\lambda) + \nt} d\lambda  \label{eq:lb-interpolation-error-L-large}
\end{align}
where the inequality is because the first integral in \eqref{eq:decomposition-integral-est-error-aliasing} is non-negative. Let $\ell' \triangleq \ell - t + 1$.  We have that
\begin{align}
&   \left[ f_{L, 0} (\lambda) \right]^2 - \left|f_{L, \ell'} (\lambda) \right|^2  \nonumber \\
 &   =  \frac{1}{L^2} \sum_{\nu = 0}^{L-1} \sum^{L-1}_{\substack{\nu' = 0, \\ \nu' \neq \nu }} \bar f_H \left( \frac{\lambda - \nu}{L} \right) \bar f_H \left( \frac{\lambda - \nu'}{L} \right)  \left[ 1  - e^{\ii 2\pi \ell' \frac{\lambda - \nu }{L}} \cdot e^{-\ii 2 \pi \ell' \frac{\lambda - \nu'}{L}}\right] \\
 &   =  \frac{2}{L^2} \sum_{\nu = 0}^{L-1} \sum^{L-1}_{\substack{ \nu' > \nu }} \bar f_H \left( \frac{\lambda - \nu}{L} \right) \bar f_H \left( \frac{\lambda - \nu'}{L} \right)  \left[1  -  \cos \left( 2\pi \ell' \frac{\nu' - \nu }{L} \right)  \right]. \label{eq:upsilon-ell}
\end{align}
Since the summands are non-negative, it follows that
\begin{align}
   \left[f_{L, 0} (\lambda) \right]^2 - \left|f_{L, \ell'} (\lambda) \right|^2   \geq   \frac{2}{L^2}  \bar f_H \left( \frac{\lambda}{L} \right) \bar f_H \left( \frac{\lambda - 1}{L} \right)  \left[1  -\cos \left(   \frac{2\pi \ell'}{L} \right) \right].
\end{align}
The RHS of \eqref{eq:lb-interpolation-error-L-large} can thus be lower-bounded as
\begin{align}
&  \int^{1/2}_{-1/2} \frac{\SNR \left( \left[ f_{L, 0} (\lambda) \right]^2 - \left|f_{L, \ell'} (\lambda) \right|^2 \right)}{\SNR f_{L,0} (\lambda) + \nt} d\lambda \nonumber \\
&  \qquad \qquad  \geq \frac{2\left[1  -  \cos \left( \frac{2\pi \ell'}{L} \right) \right]}{L^2}  \int_\mathcal{L}   \frac{\SNR \: \bar f_H \left( \frac{\lambda}{L} \right) \bar f_H \left( \frac{\lambda - 1}{L} \right)}{\SNR f_{L,0} (\lambda) + \nt}  d\lambda \label{eq:LB-second-int-result-1}
\end{align}
where $\mathcal{L}$ denotes the interval in $[-1/2, 1/2]$  where $\bar f_H \left( \frac{\lambda}{L} \right)$ and $\bar f_H \left( \frac{\lambda - 1}{L} \right)$ overlap. Note that, for $L = \frac{1}{2\lambda_D} + \varepsilon$, this interval is of Lebesgue measure 
\begin{equation}
\mu \left( \mathcal{L} \right) = \min(1, 2\lambda_D \varepsilon). \label{eq:lebesgue-mesaure-cal-L}
\end{equation}

By Fatou's lemma \cite{royden:realanalysis}, we obtain 
\begin{align}
&  \liminf_{\SNR \to \infty} \frac{2\left[1  -  \cos \left( \frac{2\pi \ell'}{L} \right) \right]}{L^2}   \int_\mathcal{L}   \frac{\SNR \: \bar f_H \left( \frac{\lambda}{L} \right) \bar f_H \left( \frac{\lambda - 1}{L} \right)}{\SNR f_{L,0} (\lambda) + \nt}  d\lambda  \nonumber \\
&  \qquad \geq \frac{2\left[1  -  \cos \left( \frac{2\pi \ell'}{L} \right) \right]}{L^2}   \int_\mathcal{L}  \liminf_{\SNR \to \infty}  \frac{\SNR \: \bar f_H \left( \frac{\lambda}{L} \right) \bar f_H \left( \frac{\lambda - 1}{L} \right)}{\SNR f_{L,0} (\lambda) + \nt}  d\lambda  \\
&  \qquad =\frac{2\left[1  -  \cos \left( \frac{2\pi \ell'}{L} \right) \right]}{L^2}  \int_\mathcal{L}   \frac{ \bar f_H \left( \frac{\lambda}{L} \right) \bar f_H \left( \frac{\lambda - 1}{L} \right)}{ f_{L,0} (\lambda) }  d\lambda. \label{eq:second-integral-error-L-ged-1overlambdaD-LB}
\end{align}
Since $\mathcal{L}$ is of positive Lebesgue measure, and since the integrand on the RHS of \eqref{eq:second-integral-error-L-ged-1overlambdaD-LB} is strictly positive, it follows from \cite{Weir_lebesgue_integration} that 
\begin{equation}
 \int_\mathcal{L}   \frac{ \bar f_H \left( \frac{\lambda}{L} \right) \bar f_H \left( \frac{\lambda - 1}{L} \right)}{ f_{L,0} (\lambda) }  d\lambda > 0. \label{eq:int-positive-lebesgue-meas}
\end{equation}
Recall that $\ell' = \ell - t + 1$. Thus, for $\ell = \nt,\dotsc,L-1$, we have 
\begin{equation}
  \cos \left(  \frac{ 2\pi \ell' }{L} \right) < 1. \label{eq:positivity-1-cos2piell}
\end{equation}
Then, combining \eqref{eq:positivity-1-cos2piell}  and  \eqref{eq:int-positive-lebesgue-meas} with \eqref{eq:second-integral-error-L-ged-1overlambdaD-LB}, \eqref{eq:LB-second-int-result-1} and \eqref{eq:lb-interpolation-error-L-large} yields
\begin{equation}
\liminf_{\SNR \to \infty} \: \: \epsilon^2_\ell (t) > 0.
\end{equation}


\end{document}